%% file: main.tex
\documentclass[11pt]{article}

\usepackage{fullpage}
\usepackage{url}
\usepackage{xspace}
\usepackage{color}
\usepackage{graphics}
\usepackage[dvips]{epsfig}

\usepackage{amsmath}
\usepackage{amssymb}
\usepackage{amsfonts}
\usepackage{graphicx}
\usepackage{algorithm}
\usepackage{algpseudocode}
\usepackage[small,compact]{titlesec}
\usepackage{times}

\newtheorem{theorem}{Theorem}[section]
\newtheorem{lemma}{Lemma}[section]
\newtheorem{corollary}{Corollary}[section]

\newtheorem{proposition}{Proposition}[section]
\newtheorem{definition}{Definition}[section]

\newcommand{\qed}{\hfill $\Box$ \bigbreak}
\newenvironment{proof}{\noindent {\bf Proof.}}{\qed}

\newcommand{\remove}[1]{}

\input{./Commun/macros.tex}

\begin{document}

	\baselineskip  0.20in 
	\parskip     0.1in 
	\parindent   0.0in 
	
	\title{{\bf Byzantine Gathering in Polynomial Time}\footnote{This work was performed within Project ESTATE (Ref. ANR-16-CE25-0009-03) and Project TOREDY. The first project is supported by French state funds managed by the ANR (Agence Nationale de la Recherche), while the second project is supported by the European Regional Development Fund (ERDF) and the Hauts-de-France region.}}
	\date{}
	\newcommand{\inst}[1]{$^{#1}$}
	
	\author{
	S\'{e}bastien Bouchard\inst{1},
	Yoann Dieudonn\'e\inst{2},
	Anissa Lamani\inst{2}\\
	\inst{1} Sorbonne Universit\'es, UPMC Univ Paris 06, CNRS, INRIA, LIP6 UMR 7606, Paris, France\\
	E-mail: \url{sebastien.bouchard@lip6.fr}\\
	\inst{2} Laboratoire MIS \& Universit\'{e} de Picardie Jules Verne, Amiens, France.\\
	E-mails: \url{{yoann.dieudonne, anissa.lamani}@u-picardie.fr}\\
	}

	\maketitle
	\vspace{-1.2cm}
	\input{./Sections/abstract.tex}


	
	
	\input{./Sections/Introduction.tex}

	\input{./Sections/Preliminaries.tex}
	\input{./Sections/buildingblock.tex}

	\input{./Sections/onlynknown.tex}
	\input{./Sections/generalresults.tex}

	\input{./Sections/Conclusion.tex}

	\bibliographystyle{plain}
	\bibliography{./Commun/main}

\end{document}

%% file: Commun/macros.tex
\makeatletter
\newcommand*{\bdiv}{%
  \nonscript\mskip-\medmuskip\mkern5mu%
  \mathbin{\operator@font div}\penalty900\mkern5mu%
  \nonscript\mskip-\medmuskip
}
\makeatother

\newcommand{\ie}{{\em i.e.,}\xspace}

\newcommand{\sta}[1]{state~#1}
\newcommand{\stA}[1]{State~#1}
\newcommand{\stb}[1]{\textbf{\stA{#1}}.}
\newcommand{\stf}[1]{{\tt #1}}
\newcommand{\stg}[1]{Algorithm~#1}
\newcommand{\sth}[1]{procedure~#1}
\newcommand{\stH}[1]{Procedure~#1}
\newcommand{\stj}[1]{subroutine~#1}
\newcommand{\stJ}[1]{Subroutine~#1}
\newcommand{\sti}[1]{{\tt #1}}

\newcommand{\std}[1]{\mathcal{O}(#1)}

\newcommand{\svA}{f}
\newcommand{\svB}{l_{min}}
\newcommand{\svC}{n}
\newcommand{\svD}{N}
\newcommand{\svE}[1]{P(#1)}
\newcommand{\svF}[1]{\svFF$(#1)$}
\newcommand{\svFF}{\sti{EXPLO}}
\newcommand{\svG}[1]{X_{#1}}
\newcommand{\svH}{\mathcal{GK}}
\newcommand{\svI}{\mathcal{T}}
\newcommand{\svJ}{1}
\newcommand{\svK}{0}
\newcommand{\svL}{{\tt TRUE}}
\newcommand{\svM}{{\tt FALSE}}
\newcommand{\svO}[1]{#1^*}
\newcommand{\svP}[1]{|#1|}
\newcommand{\svX}{bin}

\newcommand{\svQ}[1]{\lfloor #1 \rfloor}
\newcommand{\svZ}[1]{\lceil #1 \rceil}

\newcommand{\svN}{{\svZ{\log\log\svC}}}

\newcommand{\svR}[1]{(5#1 + 1) (#1 + 1) + 1}
\newcommand{\svS}[2]{max\{#2 | \svR{#2} \leq #1\}}
\newcommand{\svT}[2]{#1 - \svS{#1}{#2}}

\newcommand{\svV}[1]{\svZ{\frac {#1} {4}}}

\newcommand{\swA}{\sti{GROUP}}
\newcommand{\swAA}[3]{\swA$(#1, #2, #3)$}
\newcommand{\swAa}[1]{G_{#1}}
\newcommand{\swC}{\sti{LEARN}}
\newcommand{\swCC}{\swC$()$}
\newcommand{\swD}{\sti{MERGE}}
\newcommand{\swDD}[2]{\swD$(#1, #2)$}
\newcommand{\swDd}[1]{M_{#1}}
\newcommand{\swE}{\sti{CHECK-GATHERING}}
\newcommand{\swEE}{\swE$()$}
\newcommand{\swF}{\sti{GATHER}}
\newcommand{\swFF}[1]{\swF$(#1)$}

\newcommand{\swN}{follower}
\newcommand{\swQ}{\stf{Census}}
\newcommand{\swR}{\stf{Election}}
\newcommand{\swZ}{\stf{Synchro\-nisation}}
\newcommand{\swT}{\stf{Learning}}
\newcommand{\swU}{\stf{Optimist}}
\newcommand{\swV}{\stf{Pessimist}}

\newcommand{\szA}{\stf{Wait-for-attendees}}
\newcommand{\swO}{\stf{Invite}}
\newcommand{\swG}{\stf{Search-for-a-group}}

\newcommand{\szB}{\stf{Follow-up}}
\newcommand{\szC}{\stf{Search-for-an-invitation}}
\newcommand{\szD}{\stf{Accept-an-invitation}}
\newcommand{\szE}{\stf{Restart}}

\newcommand{\svg}{H}
\newcommand{\svh}{I}
\newcommand{\svi}{\pi}
\newcommand{\sva}{\svO{b}}
\newcommand{\svb}{{\svP{\svO{\sxb}}}}
\newcommand{\svc}{i}
\newcommand{\svd}{\gamma}
\newcommand{\sve}{\omega}
\newcommand{\svf}{\rho}
\newcommand{\svm}{z}

\newcommand{\sxA}{A}
\newcommand{\sxB}{\ell}
\newcommand{\sxC}{r}
\newcommand{\sxD}{v}
\newcommand{\sxE}{T}
\newcommand{\sxF}{t}
\newcommand{\sxa}{\sxA}
\newcommand{\sxb}{\sxB_\sxA}
\newcommand{\sxg}{\svh_1}

\newcommand{\sxj}{\svh_2}
\newcommand{\sxk}{\svh_3}

\newcommand{\sxe}{\sxC_0}
\newcommand{\sxz}{\svg_1}
\newcommand{\sya}{A}
\newcommand{\syb}{\sxC_1}
\newcommand{\syc}{\sxD_1}
\newcommand{\syd}{\sxE_1}
\newcommand{\sye}{B}
\newcommand{\syf}{\sxD_2}
\newcommand{\syg}{\sxC_2}

\newcommand{\syi}{\svg_2}
\newcommand{\syj}{\sxC_3}
\newcommand{\syk}{\sxD_3}
\newcommand{\sxl}{\sxD}
\newcommand{\sxf}{T_{\svD}}
\newcommand{\sxm}{x}
\newcommand{\sxn}{y}
\newcommand{\sxo}{p}

\newcommand{\sxq}{y}
\newcommand{\swW}{``Check-gathering''}
\newcommand{\svk}{A}
\newcommand{\svl}{B}
\newcommand{\sxr}{i}
\newcommand{\sxs}{j}
\newcommand{\sxt}{\sxD_1}
\newcommand{\suc}{\sxD_2}
\newcommand{\sxu}{\sxC_1}
\newcommand{\sae}{\sxC_2}
\newcommand{\sug}{k}
\newcommand{\szx}{g}
\newcommand{\sua}[1]{\sxF_{#1}}
\newcommand{\sxx}[2]{\sxC_{#1, #2}}
\newcommand{\sub}{\sxr}
\newcommand{\szg}{A}
\newcommand{\svj}{i}
\newcommand{\suf}{j}
\newcommand{\sxw}{\sxC_1}
\newcommand{\sxy}{A}
\newcommand{\szf}{B}
\newcommand{\sxv}{R}
\newcommand{\syl}{\sxE_1}
\newcommand{\sym}{\sxr}
\newcommand{\syn}{\sxD_1}
\newcommand{\sud}{\sxD_2}
\newcommand{\sad}{\sxC_2}
\newcommand{\syo}{j}
\newcommand{\syp}{k}
\newcommand{\syq}{A}
\newcommand{\syr}{B}
\newcommand{\sys}{\sxD_\syq}
\newcommand{\syt}{\sxD_\syr}

\newcommand{\syw}{\svm_\syq}
\newcommand{\syx}{\svm_\syr}
\newcommand{\syy}{\sxn_\syq}
\newcommand{\syz}{\sxn_\syr}
\newcommand{\sza}{\svP{\svO{\sxB_A}}}
\newcommand{\szb}{\svP{\svO{\sxB_B}}}
\newcommand{\szc}{s}
\newcommand{\sue}{\sxD_3}
\newcommand{\sxi}{\sxC_3}
\newcommand{\szd}{\sxr}
\newcommand{\sze}{A}
\newcommand{\szz}{\sxs}
\newcommand{\sra}{B}
\newcommand{\srb}{C}
\newcommand{\suh}{\sxD}
\newcommand{\sui}{\sxC}
\newcommand{\src}{\sxC_1}
\newcommand{\srd}{\sxD_1}
\newcommand{\sre}{\sxC_2}
\newcommand{\srf}{\sxD_2}
\newcommand{\szn}{\sxC_1}
\newcommand{\szo}{A}
\newcommand{\szp}{\svc_1}
\newcommand{\suj}{\svc_2}
\newcommand{\sac}{\sxC_2}
\newcommand{\szq}{\sxD_1}
\newcommand{\szr}{\sve_1}

\newcommand{\szs}{A}
\newcommand{\srg}{B}
\newcommand{\szt}{C}
\newcommand{\szu}{i}
\newcommand{\szv}{j}
\newcommand{\szh}{\sxD_\szu}
\newcommand{\szi}{\sxC_\szu}
\newcommand{\szj}{\sxC_\szv}
\newcommand{\szk}{\sxD_\szv}
\newcommand{\szl}{\sxC_1}
\newcommand{\szm}{\sxD_1}
\newcommand{\srh}{r}
\newcommand{\szw}{p}
\newcommand{\szy}{y}



%% file: Sections/abstract.tex
\begin{abstract}
Gathering a group of mobile agents is a fundamental task in the field of distributed and mobile systems. This can be made drastically more difficult to achieve when some agents are subject to faults, especially the Byzantine ones that are known as being the worst faults to handle. In this paper we study, from a deterministic point of view, the task of Byzantine gathering in a network modeled as a graph. In other words, despite the presence of Byzantine agents, all the other (good) agents, starting from {possibly} different nodes and applying the same deterministic algorithm, have to meet at the same node in finite time and stop moving. An adversary chooses the initial nodes of the agents (the number of agents may be larger than the number of nodes)  and assigns a different positive integer (called label) to each of them. Initially, each agent knows its label. The agents move in synchronous rounds and can communicate with each other only when located at the same node. Within the team, $f$ of the agents are Byzantine. A Byzantine agent acts in an unpredictable and arbitrary way. For example, it can choose an arbitrary port when it moves, can convey arbitrary information to other agents and can change its label in every round, in particular by forging the label of another agent or by creating a completely new one. 

Besides its label, which corresponds to a local knowledge, an agent is assigned some global knowledge denoted by $\mathcal{GK}$ that is common to all agents. In literature, the Byzantine gathering problem has been analyzed in arbitrary $n$-node graphs by considering the scenario when $\mathcal{GK}=(n,f)$ and the scenario when $\mathcal{GK}=f$. In the first (resp. second) scenario, it has been shown that the minimum number of good agents guaranteeing deterministic gathering of all of them is $f+1$ (resp. $f+2$). However, for both these scenarios, all the existing deterministic algorithms, whether or not they are optimal in terms of required number of good agents, have the major disadvantage of having a time complexity that is exponential in $n$ and $L$, where $L$ is the value of the largest label belonging to a good agent.

   In this paper, we seek to design a deterministic solution for Byzantine gathering that makes a concession on the proportion of Byzantine agents within the team, but that offers a significantly lower complexity. We also seek to use a global knowledge whose the length of the binary representation (that we also call size) is small. In this respect, assuming that the agents are in a \emph{strong team} i.e., a team in which the number of good agents is at least some prescribed value that is quadratic in $f$, we give positive and negative results. On the positive side, we show an algorithm that solves Byzantine gathering with strong teams in all graphs of size at most $n$, for any integers $n$ and $f$, in a time polynomial in $n$ and the length $|l_{min}|$ of the binary representation of the smallest label of a good agent. The algorithm works using a global knowledge of size $\mathcal{O}(\log \log \log n)$, which is of optimal order of magnitude in our context to reach a time complexity that is polynomial in $n$ and $|l_{min}|$. Indeed, on the negative side, we show that there is no deterministic algorithm solving Byzantine gathering with strong teams, in all graphs of size at most $n$, in a time polynomial in $n$ and $|l_{min}|$ and using a global knowledge of size $o(\log \log \log n)$.

\vspace{0.2cm}
\end{abstract}

%% file: Sections/Introduction.tex
\section{Introduction}
\subsection{Context}
Gathering a group of mobile agents is a basic problem that has been widely studied in literature dedicated to mobile and distributed systems. One of the main reasons for this popularity stems from the fact that this task turns out to be an essential prerequisite to achieve more complex cooperative works. In other words, getting fundamental results on the problem of gathering implies \emph{de facto} getting fundamental results on a large set of problems whose resolution needs to use gathering as a building block. 

The scale-up when considering numerous agents is inevitably tied to the occurrence of faults among them, the most emblematic of which is the Byzantine one. Byzantine faults are very interesting under multiple aspects, especially because the Byzantine case is the most general one, as it subsumes all the others kind of faults. In fact, in the field of fault tolerance they are considered as the worst faults that can occur. 

In this paper, we consider the problem of gathering in a deterministic way in a network modeled as a graph, wherein some agents are Byzantine. A Byzantine agent acts in an unpredictable and arbitrary manner. For instance it may choose to never stop or to never move. It may also convey arbitrary information to the other agents, impersonate the identity of another agent, and so on. In such a context, gathering is very challenging, and so far the power of such Byzantine agents has been offset by a huge complexity when solving this problem. In what follows, we seek a solution allowing to withstand Byzantine agents while keeping a ``reasonable'' complexity.

\subsection{Model and problem}
\label{sub:subm}

A team of mobile agents are initially placed by an adversary at arbitrary nodes of a network modeled as a finite, connected, undirected graph $G=(V,E)$. We assume that $|V|\leq n$. Several agents may initially share the same node and the size of the team may be larger than $n$. Two assumptions are made about the labelling of the two main components of the graph that are nodes and edges. The first assumption is that nodes are anonymous i.e., they do not have any kind of labels or identifiers allowing them to be distinguished from one another. The second assumption is that edges incident to a node $v$ are locally ordered with a fixed port numbering ranging from $0$ to $deg(v)-1$ where $deg(v)$ is the degree of $v$. Therefore, each edge has exactly two port numbers, one for each of both nodes it links. The port numbering is not supposed to be consistent{:} a given edge $(u,v)\in E$ may be the $i$-th edge of $u$ but the $j$-th edge of $v$, where $i\ne j$. These two assumptions are not fortuitous. The primary motivation of the first one is that if each node could be identified by a label, gathering would become quite easy to solve as it would be tantamount to explore the graph (via e.g. a breadth-first search) and then meet in the node having the smallest label. While the first assumption is made so as to avoid making the problem trivial, the second assumption is made in order to avoid making the problem impossible to solve. Indeed, in the absence of a way allowing an agent to distinguish locally the edges incident to a node, gathering could be proven as impossible to solve deterministically in view of the fact that some agents could be precluded from traversing some edges and visit some parts of the graph.


Time is discretized into an infinite sequence of rounds. In each round, every agent, which has been previously woken up (this notion is detailed in the next paragraph), is allowed to stay in place at its current node or to traverse an edge according to a deterministic algorithm. The algorithm is the same for all agents{:} only the input, whose nature is specified further in the subsection, varies among agents.

Before being woken up, an agent is said to be dormant. A dormant agent may be woken up only in two different ways{:} either by the adversary that wakes some of the agents at possibly different rounds, or as soon as a non-dormant agent is at the starting node of the dormant agent. We assume that the adversary wakes up at least one agent. Note that, when the adversary chooses to wake up in round $r$ a dormant agent located at a node $v$, all the dormant agents that are at node $v$ wake up in round $r$.

When an agent is woken up in a round $r$, it is told the degree of its starting node. As mentioned above, in each round $r'\geq r$, the executed algorithm can ask the agent to stay idle or to traverse an edge. In the latter case, this takes the following form{:} the algorithm asks the agent, located at node $u$, to traverse the edge having port number $i$, where $0 \leq i < deg(u)-1$. Let us denote by $(u,v)\in E$ this traversed edge. In round $r'+1$, the agent enters node $v${:} it then learns the degree $deg(v)$ as well as the local port number $j$ of $(u,v)$ at node $v$ (recall that in general $i\ne j$). An agent cannot leave any kind of tokens or markers at the nodes it visits or the edges it traverses.

In the beginning, the adversary also assigns a different positive integer (called label) to each agent. Each agent knows its label but does not know \emph{a priori} the labels of the other agents (except if some or all of them are inserted in the global knowledge $\mathcal{GK}$ that is introduced below). When several agents are at the same node $v$ in the same round $t$, they see, for each agent $x$ at node $v$, the label of agent $x$ and all information it wants to share with the others in round $t$. This transmission of information is done in a ``shouting'' mode in one round{:} all the transmitted information by all agents at node $v$ in round $t$ becomes common knowledge for agents that are currently at node $v$ in round $t$. On the other hand when two agents are not at the same node in the same round they cannot see or talk to each other{:} in particular, two agents traversing simultaneously the same edge but in opposite directions, and thus crossing each other on the same edge, do not notice this fact. In every round, the input of the algorithm executed by an agent $a$ is made up of the label of agent $a$, the up-to-date memory of what agent $a$ has seen and learnt since its waking up and some global knowledge denoted by $\mathcal{GK}$. Parameter $\mathcal{GK}$ is a piece of information that is initially given to all agents and common to all of them (i.e., $\mathcal{GK}$ is the same for all agents): its nature is precised at the end of this subsection. 
Note that in the absence of a way of distinguishing the agents, the gathering problem would have no deterministic solution in some graphs, regardless of the nature of $\mathcal{GK}$. This is especially the case in a ring in which at each node the edge going clockwise has port number $0$ and the edge going anti-clockwise has port $1${:} if all agents are woken up in the same round and start from different nodes, they will always have the same input and will always follow the same deterministic rules leading to a situation where the agents will always be at distinct nodes no matter what they do. 
 
Within the team, it is assumed that $f$ of the agents are Byzantine. A Byzantine agent has a high capacity of nuisance: it can choose an arbitrary port when it moves, can convey arbitrary information to other agents and can change its label in every round, in particular by forging the label of another agent or by creating a completely new one. All the agents that are not Byzantine are called good. We consider the task of $f$-Byzantine gathering which is stated as follows. The adversary wakes up at least one good agent and all good agents must eventually be in the same node in the same round, simultaneously declare termination and stop, despite the fact there are $f$ Byzantine agents. Regarding this task, it is worth mentioning that we cannot require the Byzantine agents to cooperate as they may always refuse to be with some agents. Thus, gathering all good agents with termination is the strongest requirement we can make in such a context. The time complexity of an algorithm solving $f$-Byzantine gathering is the number of rounds counted from the start of the earliest good agent until the task is accomplished.

We end this subsection by explaining what we mean by global knowledge, that can be viewed as a kind of advice given to all agents. Following the paradigm of algorithms with advice \cite{AbiteboulKM01,KatzKKP04,ThorupZ05,CohenFIKP08,FraigniaudIP08,FraigniaudGIP09,NisseS09}, $\mathcal{GK}$ is actually a piece of information that is initially provided to the agents at the start, by an oracle knowing the initial instance of the problem. By instance, we precisely mean: the entire graph with its port numbering, the initial positions of the agents with their labels, the $f$ agents that are Byzantine, and for each agent the round, if any, when the adversary wakes it up in case it has not been woken up before by another agent. So, for example, $\mathcal{GK}$ might correspond to the size of the network, the number of Byzantine agents, or a complete map of the network, etc. As mentionned earlier, we assume that $\mathcal{GK}$ is the same for all agents. The size of $\mathcal{GK}$ is the length of its binary representation.


\subsection{Related works}
\label{subsec:relat}

When reviewing the chronology of the works that are related to the gathering problem, it can be seen that this problem has been first studied in the particular case in which the team is made of exactly two agents. Under such a limitation, gathering is generally referred to as \emph{rendezvous}. From the first mention of the rendezvous problem in \cite{Schelling}, this problem and its generalization, gathering, have been extensively studied in a great variety of ways. Indeed, there is a lot of alternatives for the combinations we can make when addressing the problem, e.g., by playing on the environment in which the agents are supposed to evolve, the way of applying the sequences of instructions (i.e., deterministic or randomized) or the ability to leave some traces in the visited locations, etc. In this paper, we are naturally closer to the research works that are related to deterministic gathering in networks modeled as graphs. Hence, we will mostly dwell on this scenario in the rest of this subsection. However, for the curious reader wishing to consider the matter in greater depth, we invite him to consult \cite{CieliebakFPS12,AgmonP06,IzumiSKIDWY12} that address the problem in the plane via various scenarios, especially in a system affected by the occurrence of faults or inaccuracies for the last two references. Regarding randomized rendezvous, a good starting point is to go through \cite{Alpern02,Alpern03,KranakisKR06}.

Now, let us focus on the area that concerns the present paper most directly, namely deterministic rendezvous and/or gathering in graphs. In most papers on rendezvous in networks, a synchronous scenario was assumed, in which agents navigate in the network in synchronous rounds. Under this context, a lot of effort has been dedicated to the study of the feasibility and to the time (i.e., number of rounds) required to achieve the task, when feasible. For instance, in \cite{DessmarkFKP06} the authors show a rendezvous algorithm polynomial in the size of the graph, in the length of the shorter label and in the delay between the starting time of the agents. In \cite{KowalskiM08} and \cite{Ta-ShmaZ14} solutions are given for rendezvous, which are polynomial in the first two of these parameters and independent of the delay. While these algorithms ensure rendezvous in polynomial time (i.e., a polynomial number of rounds), they also ensure it at polynomial cost where the cost corresponds here to the total number of edge traversals made by both agents until meeting. Indeed, since each agent can make at most one edge traversal per round, a polynomial time always implies a polynomial cost. However, the reciprocal may be not true, for instance when using an algorithm relying on a technique similar to ``coding by silence'' in the time-slice algorithm for leader election \cite{Lynch96}:
``most of the time'' both agents stay idle, in order to guarantee that agents rarely move simultaneously. Thus these parameters of cost and time are not always linked to each other. This was recently highlighted in \cite{MillerP16} where the authors studied the tradeoffs between cost and time for the deterministic rendezvous problem. Some other efforts have been also dedicated to analyse the impact on time complexity of rendezvous when in every round the agents are brought with some pieces of information by making a query to some device or some oracle, see, e.g., \cite{DKU14,MillerP15}. Along with the works aiming at optimizing the parameters of time and/or cost of rendezvous, some other works have examined the amount of memory that is required to achieve deterministic rendezvous e.g., in \cite{FraigniaudP08,FraigniaudP13} for tree networks and in \cite{CzyzowiczKP12} for general networks. 

Apart from the synchronous scenario, the academic literature also contains several studies focusing on a scenario in which the agents move at constant speed, which are different from each other, or even move asynchronously{:} in this latter case the speed of agents may then vary and is controlled by the adversary. For more details about rendezvous under such a context, the reader is referred to \cite{MarcoGKKPV06,CzyzowiczPL12,GuilbaultP13,DieudonnePV15,KranakisKMPR17} for rendezvous in finite graphs and \cite{BampasCGIL10,CollinsCGL10} for rendezvous in infinite grids.

As stated in the previous subsection, our paper is also related to the field of fault tolerance since some agents may be prone to Byzantine faults. First introduced in \cite{PeaseSL80}, a Byzantine fault is an arbitrary fault occurring in an unpredictable way during the execution of a protocol. Due to its arbitrary nature, such a fault is considered as the worst fault that can occur. Byzantine faults have been extensively studied for ``classical'' networks i.e., in which the entities are fixed nodes of the graph (cf., e.g., the book \cite{Lynch96} or the survey \cite{BarborakM93}). To a lesser extend, the occurrence of Byzantine faults has been also studied in the context of mobile entities evolving on a one-dimensional or two-dimensional space, cf. \cite{AgmonP06,DefagoGMP06,CzyzowiczGKKNOS16}. 

Gathering in arbitrary graphs in presence of many Byzantine agents was considered in \cite{DieudonnePP14,BouchardDD16}. Actually, our model is borrowed from both these papers, and thus they are naturally the closest works to ours. In \cite{DieudonnePP14}, the problem is introduced via the following question: {\it what is the minimum number $\mathcal{M}$ of good agents that guarantees f-Byzantine gathering in all graphs of size $n$?} In~\cite{DieudonnePP14}, the authors provided several answers to this problem by firstly considering a relaxed variant, in which the Byzantine agents cannot lie about their labels, and then by considering a harsher form (the same as in our present paper) in which Byzantine agents can lie about their identities. For the relaxed variant, it has been proven that the minimum number $\mathcal{M}$ of good agents that guarantees $f$-Byzantine gathering is precisely $1$ when $\mathcal{GK}=(n,f)$ and $f+2$ when $\mathcal{GK}$ is reduced to $f$ only. The proof that both these values are enough, relies on polynomial algorithms using a mechanism of blacklists that are, informally speaking, lists of labels corresponding to agents having exhibited an ``inconsistent'' behavior. Of course, such blacklists cannot be used when the Byzantine agents can change their labels and in particular steal the identities of good agents. Still in \cite{DieudonnePP14}, the authors give for the harsher form of $f$-byzantine gathering a lower bound of $f+1$ (resp. $f+2$) on $\mathcal{M}$ and a deterministic gathering algorithm requiring at least $2f+1$ (resp. $4f+2$) good agents, when $\mathcal{GK}=(n,f)$ (resp. $\mathcal{GK}=f$). Both these algorithms have a huge complexity as they are exponential in $n$ and $L$, where $L$ is the largest label of a good agent evolving in the graph. 
Some advances are made in \cite{BouchardDD16}, via the design of an algorithm for the case $\mathcal{GK}=(n,f)$ (resp. $\mathcal{GK}=f$)  that works with a number of good agents that perfectly matches the lower bound of $f+1$ (resp. $f+2$) shown in \cite{DieudonnePP14}. However, these algorithms also suffer from a complexity that is exponential in $n$ and $L$.

\subsection{Our results}

As mentioned just above, the existing deterministic algorithms dedicated to $f$-Byzantine gathering all have the major disadvantage of having a time complexity that is exponential in $n$ and $L$, when Byzantine agents are allowed to change their labels. Actually, these solutions are all based on a common strategy that consists in enumerating the possible initial configurations, and successively testing them one by one. Once the testing reaches the correct initial configuration, the gathering can be achieved. However, in order to get a significantly more efficient algorithm, such a costly strategy must be abandoned in favor of a completely new one.

   In this paper, we seek to design a deterministic solution for Byzantine gathering that makes a concession on the proportion of Byzantine agents within the team, but that offers a signifi\-cantly lower complexity. We also seek to use a global knowledge whose the length of the binary representation (that we also call size) is small. In this respect, assuming that the agents are in a \emph{strong team} i.e., a team in which the number of good agents is at least the quadratic value $5f^2+6f+2$, we give positive and negative results. On the positive side, we show an algorithm that solves $f$-Byzantine gathering with strong teams in all graphs of size at most $n$, for any integers $n$ and $f$, in a time polynomial in $n$ and $|l_{min}|$. The algorithm works using a global knowledge of size $\mathcal{O}(\log \log \log n)$, which is of optimal order of magnitude in our context to reach a time complexity that is polynomial in $n$ and $|l_{min}|$. Indeed, on the negative side, we show that there is no deterministic algorithm solving $f$-Byzantine gathering with strong teams, in all graphs of size at most $n$, in a time polynomial in $n$ and $|l_{min}|$ and using a global knowledge of size $o(\log \log \log n)$. 

  \subsection{Roadmap}

The next section is dedicated to the presentation of some basic definitions and routines that we need in the rest of this paper. In Section~\ref{sec:block}, we describe two building blocks that are used in turn in Section~\ref{sec:onlyn} to establish our positive result. In Section~\ref{sec:gen}, we prove our negative result. Finally we make some concluding remarks in Section~\ref{sec:ccl}.

%% file: Sections/Preliminaries.tex
\section{Preliminaries}
\label{sec:pre}
Throughout the paper, $\log$ denotes the binary logarithm. An agent will be designated by a capitalized letter, and the label of an agent $X$ will be denoted by $\ell_X$. The length of the binary representation of $\ell_X$ will be denoted by $|\ell_X|$. The length of the binary representation of the smallest label of a good agent in a given team will be denoted by $|l_{min}|$.

Several routines given in this paper will use a procedure whose aim is graph exploration,  i.e., visiting all nodes of the graph. This procedure, based on universal exploration sequences (UXS), is a corollary of the result of Reingold \cite{Reingold08}. Given any positive integer $\svC$, this procedure, called \svF{\svC}, allows the executing agent to traverse all nodes of any graph of size at most $\svC$,
starting from any node of this graph, using $\svE{\svC}$ edge traversals, where $P$ is some polynomial. After entering a node of degree $d$ by some port $p$,
the agent can compute the port $q$ by which it has to exit; more precisely $q=(p+x_i)\bmod d$, where $x_i$ is the corresponding term of the UXS of length $\svE{\svC}$. We denote by $\svG{\svC}$ the execution time of procedure \svFF{} with parameter $\svC$ (note that $\svG{\svC}=\svE{\svC}+1$).

Besides this exploration procedure, we will use a label transformation derived from~\cite{DessmarkFKP06}. Let $\ell_B$ be the label of an agent $B$ and $b_1 \ldots b_c$ its binary representation with $c$ its length. The binary representation of the corresponding transformed label $\ell_B^*$ is $10 b_1b_1 \ldots b_cb_c 01 10 b_1b_1 \ldots b_cb_c 01$. This transformation is made to ensure the following property that is used in the proof of correctness of our algorithm in Section~\ref{sec:onlyn}.

\begin{proposition} \label{prelim:label}
	Let $\ell_B$ and $\ell_X$ be two labels such that $\ell_B < \ell_X$. Let $b^*_1b^*_2 \ldots b^*_{4c+8}$ and $x^*_1x^*_2 \ldots x^*_{4y+8}$ be the respective binary representations of $\ell_B^*$ and $\ell_X^*$, with $c$ and $y$ the lengths of the binary representations of $\ell_B$ and $\ell_X$ respectively. There exist two positive integers $i \leq 2c+4$ and $2c+4 < j \leq 4c+8$ such that $b^*_i \ne x^*_i$ and $b^*_j \ne x^*_j$.
\end{proposition}

\begin{proof}
	Let $b_1b_2 \ldots b_c$ (resp. $x_1x_2\ldots x_y$) the binary representation of $\ell_B$ (resp. $\ell_X$). There are two cases to consider: either $c = y$ or $c < y$. In the first case, since $\ell_B \ne \ell_X$, there exists a positive integer $i \leq c$ such that $b_i \ne x_i$. This implies in particular that $x^*_{2i+1} \ne b^*_{2i+1}$ and $x^*_{2c+2i+5} \ne b^*_{2c+2i+5}$. In the second case, either $x^*_{2c+3} \ne b^*_{2c+3}$ or $x^*_{2c+4} \ne b^*_{2c+4}$, and if $x^*_{2c+5} = b^*_{2c+5}$ then $x^*_{2c+6} \ne b^*_{2c+6}$. Hence, in each case the proposition holds.
\end{proof}


Throughout the paper, we will recurrently design some routines in the form of a description of several states, where an agent has to apply specific rules, along with how to transit among them. In each round spent executing such a routine, we assume that a good agent will tell its current state to the other agents sharing the same node. Sometimes, we will require that an agent also tells extra information other than only its state: when such a situation arises, we will obviously precise this point.
Moreover, in the description of our states, we will use different expressions that are as follows. When we say ``agent $A$ enters state W'', we precisely mean that at the previous round, agent $A$ was in some state U $\neq$ W and at the current round, it is in state W. When we say ``agent $A$ exits state X'', we mean agent $A$ remains in state X until the end of the current round and is in some state V $\neq$ X at the following round. Lastly, when we say ``agent $A$ transits from state Y to state Z'', we mean agent $A$ exits state Y at the current round and enters state Z at the following one. Thus, in each round, agent $A$ is always exactly in at most one state.



	
	
	


%% file: Sections/buildingblock.tex
\section{Building blocks} \label{sec:block}

	To design our solution that is given in Section~\ref{sec:onlyn}, we need to describe two prior subroutines that will be used as building blocks.

	In the rest of this section, for each of both building blocks, we first explain the high level idea that is behind it. Then, we give a detailed description of it. Finally, we show its correctness and analyze its time complexity.

	\subsection{\stH{\swA}}

		The first building block called \swA{} takes as input three integers $\svI$, $n$ and $\svX$ such that $\svX\in\{0;1\}$. Let $x$ be an integer that is at least $f+2$. Roughly speaking, \stj{\swAA{\svI}{n}{\svX}} ensures that $(x-f)$ good agents finish the execution of the subroutine at the same round and in the same node in a graph of size at most $n$ provided the following two conditions are verified: the number of agents is at least $(x-1)(f+1)+1$, and all good agents start executing the subroutine in some interval lasting at most $\svI$ rounds, with the same parameters except for the last one that has to be $0$ (resp. $1$) for at least one good agent. The time complexity of the procedure is polynomial in the first two parameters $\svI$ and $n$.

		\subsubsection{High level idea}

			As mentioned previously, \stj{\swA} aims at ensuring that $x-f$ good agents finish the execution of the subroutine at the same round and in the same node. To achieve this, we have to face several difficulties, especially the fact that the agents know neither $x$ nor $f$, and also the fact that agents have a priori no mean to detect whether an agent is good or not. Indeed, we cannot have instructions like ``If there are at least $x-f$ good agents in my current node, then\dots''. We cannot even have ``If there are at least $x$ or $f$ agents in my current node, then...'' no matter whether there are some Byzantine agents or not in the current node. So, to circumvent these problems, \sth{\swA} is made of two phases. The first phase aims at ensuring that at least $x$ agents executing the first phase meet in the same node (even though the involved agents do not detect this event). This phase lasts exactly the same time for each good agent and when it finishes it, a good agent is at the node from which it started executing it. The second phase consists, for a good agent, in replaying in the same order the same edge traversals and waiting periods made during the first $i$ rounds of its first phase started at round $t$, such that $t+i$ is the round when the agent was with the maximal number of agents executing the first phase (if there are several such rounds, we choose the latest one). Once this is done, the agent stops executing \swA{}. By doing so, we have the guarantee that $x-f$ good agents (those involved in the last maximal meeting of the first phase) will stop executing the second phase (and thus \sth{\swA}) in the same node and at the same round, as all the meetings involving the maximal number of agents in the first phase, necessarily involve at least $x$ agents. Hence, the key of the procedure is to make $x$ agents meet in the first phase.  

During the execution of the first phase, the agents are partitioned into two distinct groups, namely \emph{followers} and \emph{searchers}. The first group corresponds to agents executing the subroutine with $\svX=0$ and the second group corresponds to those executing it with $\svX=1$. 

The first phase works in steps $1,2,\ldots,\mathcal{S}$ where $\mathcal{S}$ is some polynomial in $\svI$ and $n$. At a very high level, in each step, the main role of followers is to remain idle in their initial starting nodes in order to ``mark" possible positions on which $x$ agents could meet, while the main role of searchers is to look for these positions. To this end, each searcher will make use of a kind of map that it initially computes during the first step by making an entire traversal of the graph, using \sth{\svF{\svC}}. Actually, this map corresponds to a sequence $P$ of objects symbolizing every visited node $v$ along with the list of labels of the agents that are (or pretend to be) followers present in node $v$ at the time of the visit by the searcher. More precisely, the length of $P$ is equal to the number of visited nodes in \svF{\svC}, and the $i$-th object of $P$ contains, among other information, the set of all followers' labels present in the $i$-th visited node of the traversal. Note that such a map will be called \emph{imperfect map} as some nodes can be represented several times in the sequence $P$. Indeed \svF{\svC} guarantees that each node is visited at least once but some nodes may be visited more than once. The use of the qualifying term ``imperfect'' also stems from the fact that the list of followers' labels that are stored in $P$ may be plagued by artificial ones created by Byzantine agents. In all the other steps, the searchers never recompute a new imperfect map, but always use the one computed in the first step, along with some possible updates on the lists of labels. How and when these updates are applied is explained below: they are obviously related to ``bad behaviors'' coming from Byzantine agents.

For the convenience of the explanation, let us first consider an 
ideal situation in which there is a unique follower among the good agents. If there is no Byzantine agent, during the first step, by moving to the node that hosts the unique follower, all searchers meet in the same node: using their maps, they are all able to determine a path to this follower. Thus, if the number of good agents is at least $x$, there is necessarily a round in which $x$ agents meet in the same node. However, when Byzantine agents come into the picture, the problem becomes a tricky one, as these malicious agents can also pretend to have the status of followers (with the same label or not). Hence, all the searchers may not necessarily choose to move towards the same follower, which may prevent \emph{in fine} the meeting of $x$ agents. To deal with this issue and limit the confusion caused by Byzantine agents, in each step every good searcher $A$ proceeds as follows. Let $\ell$ be the smallest label in sequence $P$ (corresponding to the imperfect map of $A$) and let $i$ be the first object of $P$ in which $\ell$ appears. Agent $A$ moves to the node $u$ of the graph corresponding to the $i$-th object of $P$ in order to meet again the follower with label $\ell$ and waits some prescribed amount of rounds with it at node $u$. If $A$ does not see a follower with label $\ell$ when reaching $u$, or at some point during its waiting period at $u$ it does not see anymore any follower with label $\ell$, then agent $A$ updates its imperfect map by removing $\ell$ from the list of the $i$-th object. Then, in all cases, the agent will end up starting the following step (if any) with its possibly updated map. The total number $\mathcal{S}$ of steps has been carefully chosen so that it is larger than the total number of map updates that can be made by all good agents in the network. Hence, we can ensure the existence of a step in which there is no map update: in such a step we will be able to prove that the number of different locations that are reached by searchers is at most $f+1$. Thus, if the number of good agents is at least $(x-1)(f+1)+1$, using arguments relying on the \textit{pigeonhole principle}, we will prove the meeting of $x$ agents. Keep in mind that all the above explanations are made under the assumption there is a single good follower in the team. When there are more than one good follower, things get more complicated. For instance, observe in this case that even though the number of good agents is at least $(x-1)(f+1)+1$, our approach, without additional precautions, may fail to make at least $x$ agents meet on the same node as the number of good searchers may not be enough to ensure the meeting. Indeed, for a given number of agents, the more followers, the less searchers to distribute. However, through extra technical actions requiring sometimes some followers to end up behaving as a searcher, we will be able to overcome this issue and still ensure the meeting of $x$ agents provided the cardinality of the set of good agents is at least $(x-1)(f+1)+1$.

\subsubsection{Detailed description}

 
To describe \stj{\swA}, we use a function called \textit{IM} that takes as input two integers $n$ and $q \in \{0, 1\}$, and returns an ordered sequence $P$ of lists of labels: $P=<L_{1}, \dots, L_{\svG{\svC}}>$. The returned sequence $P$ is called an \emph{imperfect map}. When a given agent $A$ performs $IM(q,n)$, it actually executes \svF{\svC} with some additional actions. At each step of \svF{\svC}, depending on the value of $q$, $A$ checks the presence of a given agent or a group of agents to compute $P$. During the first step, the agent is at the node from which it starts $IM(q,n)$. Let us consider the $j^{th}$ step of \svF{\svC} ($j \in \{1, \dots , \svG{\svC}\})$ and let $u$ be the node on which $A$ is at this step. If $q=0$, $L_j$ is a list of pairwise distinct labels such that $\ell_B\in L_j$ iff there is on $u$ an agent $B$ with label $\ell_B$ being or pretending to be a follower. If $q=1$, $L_j$ is a list of pairwise distinct labels such that $\ell_B\in L_j$ iff there is on $u$ an agent $B$ with label $\ell_B$ being or pretending to be a follower in \sta{\szA}. \stA{\szA} is defined in the description of the algorithm. When $\ell_B$ is added to a list of $P$ by $A$, we say that $A$ records $B$. At the end of \svF{\svC}, the agent traverses all the edges traversed in \svF{\svC} in the reverse order, and then it exits $IM(q,n)$.


To facilitate the presentation of the formal description of \sth{\swA}, we also need the following two definitions.

\begin{definition}[Useful map]
\label{def:use}
An imperfect map $P$ is said to be useful iff $P$ contains a non-empty list.
\end{definition}

\begin{definition}[Index of a map]
\label{def:index}
Let $P=<L_{1}, \dots, L_{\svG{\svC}}>$ be a useful map. Let $S$ be the set of every label that appears in at least one list of $P$. Let $j$ be the smallest integer such that $L_j$ contains the smallest label of $S$: $j$ is the index of $P$.
\end{definition}



Now, we are ready to give the formal description of the subroutine. \stJ{\swAA{\svI}{n}{\svX}} comprises two phases: \textbf{Process} and \textbf{Build-up}. Let us consider a given agent $A$ executing \swAA{\svI}{n}{\svX} from an initial node $v$. When $\svX=0$, the agent is said to be a \textit{follower}. Otherwise, it is said to be a \textit{searcher}. 
The description is in the form of several states along with rules to transit among them. At the beginning of each state, the agent is in its initial node $v$. 

\begin{itemize}
\item  \textbf{Phase Process}. Agent $A$ proceeds in steps $1, 2, \dots, \mathcal{S}$ where $\mathcal{S}=n^2.\svI.\svG{\svC}+1$. Assume without loss of generality that $A$ is at step $s \in \{1, \dots, \mathcal{S}\}$. Unless stated explicitly, all the transitions between states which are presented below are performed within the same step. In all what follows {$\mathcal{H}=(n+1)[\svI+4\svG{\svC}+(\svG{\svC}. n)(\svI\svC+\svC)(2\svG{\svC}+\svI)]+3$}. We describe $A$'s behavior depending on the value of $bin$.

\begin{itemize}

\item $bin=0$ ($A$ is a follower). In this case, $A$ can be in one of the following states: \swO{}, \szA{}, \swG{} and \szB{}. At the beginning of each step $s$, agent $A$ is in \sta{\swO}. The actions to be performed in each state are presented in what follows. \\

\stb{\swO} Agent $A$ waits $2\svI+3\svG{\svC}$ rounds.  At the end of this waiting time, if $A$ is on the same node as at least one searcher, $A$ transits to \sta{\szA}. Otherwise, it transits to \sta{\swG}.\\

\stb{\szA} Agent $A$ waits $2\svI+ \svG{\svC} + \mathcal{H}$ rounds. If at each round of this waiting period, there is at least one searcher at node $v$, then at the end of the waiting period agent $A$ transits to \sta{\szB}. Otherwise, as soon as there is a round of the waiting period when there is no searcher at node $v$, agent $A$ transits to \sta{\swG} (hence the waiting period may be prematurely stopped).\\


\stb{\swG} Let $k$ be the number of rounds spent by agent $A$ in \\ \sta{\szA} of step $s$. 

Note that $k=0$ if $A$ transited directly to \sta{\swG} from \sta{\swO} in step $s$. Let $w$ be a counter, the initial value of which is $0$. The way this counter is incremented and decremented is explained below. 

While agent $A$ does not reach round $t+2\svI+ \svG{\svC} + \mathcal{H}-k$ where $t$ is the round when it entered this state in step $s$, it proceeds as follows (thus, what follows is then interrupted when reaching round $t+2\svI+ \svG{\svC} + \mathcal{H}-k$). Agent $A$ first waits $\svI$ rounds and then executes $IM(1,n)$. Once this is done, the agent has a map $P$. Each time $P$ is useful (refer to Definition~\ref{def:use}) and $w=0$, the agent performs the first $i-1$ edge traversals of \svF{\svC} from its initial node $v$ where $i$ is the index of $P$: just before each edge traversal, counter $w$ is incremented by one. Let us refer to the node reached at the end of these $i-1$ edge traversals by $u$. As long as there is a follower $B$ in \sta{\szA} on $u$ such that $\ell_B$ is the smallest label in the $i$-th list $L_i$ of $P$, $A$ remains idle. By contrast, if there is no such follower on $u$ in some round, agent $A$ updates $P$ by removing from $L_i$ its smallest element and then goes back to its initial node $v$ by performing the $(i-1)$ edge traversals executed above in the reverse order: just before each edge traversal of this backtrack, counter $w$ is decremented by one.

As soon as agent reaches round $t+2\svI+ \svG{\svC} + \mathcal{H}-k$, the agents proceeds as follows: if $w=0$, it transits to \sta{\szB}. Otherwise, if $w>0$, $A$ goes back to its initial node $v$ by traversing in the reverse order the sequence of $w$ edges $e_1,e_2,\ldots,e_k$ corresponding to the  $w$ first edge traversals of \svF{\svC} from node $v$: once this is done, it transits to \sta{\szB}.\\

\stb{\szB}
Let $x$ be the number of rounds elapsed from the beginning of the current step. Agent $A$ waits $5\svI+5\svG{\svC}+ \mathcal{H}-x$ 
 rounds. At the end of the waiting time, if $s<\mathcal{S}$, $A$ transits to \sta{\swO} of step $s+1$. Otherwise, $A$ transits to \sta{\szE} of phase \textbf{Build-up}. \\

\item $bin=1$ ($A$ is a searcher). Agent $A$ can be in one of the following states: 

\szC{}, \szD{} and \szB{}.

At the beginning of each step $s$, agent $A$ is in \sta{\szC}. We present in what follows the set of actions to be performed for each state.  \\

\stb{\szC} Agent $A$ first waits $\svI$ rounds. Next, if $s=1$ (first step of phase \textbf{Process}), $A$ executes $IM(0,\svC)$ and then transits to \sta{\szD}. The output of the execution of $IM(0,n)$ is stored in variable $Z$. This variable may be updated in the current step as well as the following ones: each time we will mention this variable, we will implicitly consider its up-to-date value. If $s>1$, $A$ waits $2\svG{\svC}$ rounds and then transits to \sta{\szD}. \\

\stb{\szD} In the case where $Z$ is not a useful map, $A$ transits to \sta{\szB}.
Otherwise, let $j$ and $\ell$ be the index of $Z$ and the smallest label in the $j$-th list of $Z$ respectively. Agent $A$  performs the first $j-1$ edge traversals of \svF{\svC}. Let $t$ be the round when agent $A$ finishes these first $j-1$ edge traversals, and let $u$ be the node reached by $A$ in round $t$. As soon as there is a round in $\{t+1,t+2,\ldots, t+2\svI+\svG{\svC}+ \mathcal{H}\}$ for which there is no follower $B$ at node $u$ such that label $\ell_B=\ell$, agent $A$ updates $P$ by removing $\ell$ from $L_j$ and goes back to its initial node $v$ by performing the $(j-1)$ edge traversals executed above in the reverse order. Once this backtrack is done, agent $A$ transits to \sta{\szB}.

If agent $A$ is still in \sta{\szD} in round $t+2\svI+\svG{\svC}+ \mathcal{H}$, 
it goes back to its initial node $v$ by performing the $(j-1)$ edge traversals executed above in the reverse order, and then it transits to \sta{\szB} (note that in this latter case, $Z$ remains unchanged).\\

\stb{\szB} Let $x$ be the number of rounds elapsed from the beginning of the current step. Agent $A$ waits 
$5\svI+5\svG{\svC}+ \mathcal{H}-x$ 
 rounds. At the end of the waiting time, if $s< \mathcal{S}$, then $A$ transits to \sta{\szC} of step $s+1$.  Otherwise, it transits to \sta{\szE} of phase \textbf{Build-up}.

\end{itemize}
\item \textbf{Phase Build-up}. Agent $A$ can only be in \sta{\szE}.

\emph{//At the beginning of this phase, the agent is at the node from which it started \sth{\swA} i.e., node $v$}

\stb{\szE} Let $r$ be the round in which $A$ initiated \swA{} and let $r+i$ be the round in phase \textbf{Process} in which $A$ is on a node containing the largest number of agents (including $A$ itself) that are not in \sta{\szE}. If there are several such rounds, it chooses the one with the largest value $i$. Denote by $r'$ the round in which the agent enters this state. From round $r'$ to $r'+i-1$, agent $A$ replays exactly the same waiting periods and edges traversals from round $r$ to $r+i-1$. More precisely, for each integer $y$ in $\{0,1,\ldots,i-1\}$, if agent $A$ remains idle (resp. leaves the current node via a port $o$) from round $r+y$ to round $r+y+1$, then agent $A$ remains idle (resp. leaves the current node via port $o$) from round $r'+y$ to $r'+y+1$.
In round $r'+i$, the agent stops the execution of \swA{}.
\end{itemize}

		\subsubsection{Correctness and complexity analysis} \label{sub:proof1}
		

 Let $\mathcal{E}$ and $\Delta$ be respectively the set of all good agents in the network and the first round in which an agent of $\mathcal{E}$ starts executing \swAA{\svI}{n}{\svX}. Let $x$ be an integer that is at least $f+2$. To conduct the proof of correctness as well as the complexity analysis, we assume in the rest of this subsection that $|\mathcal{E}|\geq (x-1)(f+1)+1$, every agent of $\mathcal{E}$ starts executing \swAA{\svI}{n}{\svX} at round $\Delta+\svI-1$ at the latest, and at least one agent of $A$ starts executing the procedure with $\svX=0$ (resp. $\svX=1$).




We start with the following lemma about the duration of each step and the duration of phase Process. Recall that $\mathcal{S}$ and $\mathcal{H}$ are polynomials in $n$ and $\svI$ given in the detailed description of \sth{\swA}.

\begin{lemma}
\label{lem:duration}
Let $A$ be an agent of $\mathcal{E}$. We have the following two properties.
\begin{enumerate}
\item Each step of phase Process executed by $A$ lasts exactly $5\svI+5\svG{\svC}+ \mathcal{H}$ rounds.
\item The execution of phase Process by agent $A$ lasts exactly $\mathcal{S}\cdot(5\svI+5\svG{\svC}+\mathcal{H})$ rounds. 
\end{enumerate}
\end{lemma}

\begin{proof}
According to the algorithm, $\mathcal{S}$ corresponds to the number of steps in phase Process. So, if the first property holds, the second one also holds. Hence, to prove the lemma, it is enough to prove that the first property is true: this will be the purpose of the rest of this proof.

Let $s$ be a step of phase Process executed by agent $A$. Let us first prove that $A$ transits to \sta{\szB} of step $s$ after having spent at most $4\svI+5\svG{\svC}+\mathcal{H}$ rounds in this step. Depending on the value of $bin$, $A$ can be either a searcher or a follower. We consider the two cases.

\begin{itemize}
\item $A$ is a follower. The state of $A$ in the first round of every step $s$ of phase Process during the execution of \swA{} is \swO{}. Agent $A$ spends $2\svI+3\svG{\svC}$ rounds in \sta{\swO} before transiting to either \sta{\szA} or \sta{\swG} depending on whether there is a searcher on the same node as $A$ at the end of this waiting time. Agent $A$ remains in either \sta{\szA} or \swG{} at most $2\svI+ 2\svG{\svC} + \mathcal{H}$ rounds before transiting to \sta{\szB}. Hence, agent $A$ spends at most $4\svI+5\svG{\svC}+\mathcal{H}$ in step $s$ before transiting to \sta{\szB}.  

\item $A$ is a searcher. The state of $A$ in the first round of \swA{} is \szC{}. First, agent $A$ waits $\svI$ rounds. Next, if $s=1$, $A$ executes $IM(0,n)$ that lasts $2\svG{\svC}$ rounds before transiting to \sta{\szD}. Otherwise $s>1$ and $A$ waits $2\svG{\svC}$ rounds before transiting to \sta{\szD}. That is, in both cases, $A$ spends $\svI+2\svG{\svC}$ in total before transiting to \sta{\szD}. Once $A$ transits to \sta{\szD}, if $Z$, the output of $IM(0,n)$ performed while in \sta{\szC} in the first step of phase Process, is not a useful map, $A$ transits to \sta{\szB} and the lemma holds. If by contrast, $Z$ is useful then $A$ performs the first $(j-1)$ edge traversals of \svF{\svC} where $j$ is the index of $Z$. Agent $A$ waits at most
$2\svI+\svG{\svC}+ \mathcal{H}$ rounds (with a follower) and then performs less than $\svG{\svC}$ edge traversals to retrieve its initial position before transiting to \sta{\szB}. So, agent $A$ spends at most $3\svI+5\svG{\svC}+\mathcal{H}$ in step $s$ before transiting to \sta{\szB}.
 
\end{itemize}

Hence, whether $A$ is a follower or not, it spends at most $x\leq 4\svI+5\svG{\svC}+\mathcal{H}$ rounds in step $s$ before transiting to \sta{\szB} of step $s$. However, according to \sta{\szB}, the agent waits exactly $5\svI+5\svG{\svC}+ \mathcal{H}-x$ rounds before leaving step $s$. Hence, the lemma holds.
\end{proof}

From the previous lemma, we get the following corollary and remark.

\begin{corollary}\label{block:cor:delay}
Let $A$ and $B$ be any two good agents of $\mathcal{E}$ such that $t_A-t_B\geq0$, where $t_A$ (resp. $t_B$) is the round when $A$ (resp. $B$) starts executing \swA{}.
For every step $s$ of phase Process, agent $A$ finishes executing $s$, exactly $t_A-t_B$ rounds after $B$ finishes executing it.
\end{corollary}


In the following, by initial node we mean the node from which the agent starts executing \sth{\swA}. Note that the statement of Lemma~\ref{block:lem:snapshot} calls for the notion of ``recording'' that is introduced in the description of function $IM$.

\begin{lemma}\label{block:lem:snapshot}
	Let $A$ be a good searcher of $\mathcal{E}$. For every follower $B$ in $\mathcal{E}$, agent $A$ records $B$ during its execution of $IM(0,n)$ when $B$ is on its initial node.
\end{lemma}

\begin{proof}
	According to the algorithm, agent $A$ executes $IM(0,n)$ while in \sta{\szC} of step $1$. More precisely, when starting step $1$, agent $A$ first waits $\svI$ rounds and then executes $IM(0,n)$ that lasts $2\svG{\svC}$ rounds. On the other hand, agent $B$ waits $2\svI+3\svG{\svC}$ rounds in its initial node at the beginning of step $1$. Hence, in view of the initial delay $\svI$, the lemma follows.
\end{proof}

To continue, we need to introduce the definition of \emph{target node}. A node $u$ is said to be a target node of a good searcher $A$ in a step $s>1$, if $u$ is the node that is reached after performing the first $(j-1)$ edge traversals of \svF{\svC} from the initial node of $A$ and $j$ is the index of the imperfect map of $A$ at the beginning of its execution of step $s$.

\begin{lemma}\label{Block:lem:together}
Let $A$ be a searcher of set $\mathcal{E}$ starting a step $s$ with a useful map $P$, the index of which is $j$. Let $\ell$ be the smallest label in the $j$-th list of $P$. If the target node of $A$ in step $s$ is the initial node of a good follower $B$ such that $\ell_B=\ell$, then $A$ does not update $P$ in any step $s'\geq s$.
\end{lemma}

\begin{proof}
We prove by induction on $i\geq0$ that $A$ does not update $P$ in step $s+i$. First consider, the initial step in which $i=0$.
Assume by contradiction that the target node of $A$ in step $s$ is the initial node $u$ of a good follower $B$ such that $\ell_B=\ell$, but $A$ updates $P$ in step $s$. According to \sth{\swA}, agent $A$ reaches node $u$ while in \sta{\szD}. Then, agent $A$ updates $P$ in step $s$ only if agent $A$ does not meet agent $B$ when reaching target node $u$ or $A$ notices the absence of $B$ on $u$ within  $2\svI+ \svG{\svC}+\mathcal{H}$ rounds after its meeting with $B$. However, when agent $B$ starts step $s$, it first waits $2\svI+ 3\svG{\svC}$ in \sta{\swO}. Hence in view of Corollary~\ref{block:cor:delay} and the definition of $\svI$, agent $A$ meets $B$ when reaching target node $u$ while $B$ is in \sta{\swO}: indeed agent $A$ spends {$\svI+2\svG{\svC}$} rounds in \sta{\szC} and at most $\svG{\svC}$ rounds in \sta{\szD} before reaching its target node $u$ at some round $t$. Moreover, since agent $A$ and $B$ are good, according to states \szA{} and \szD{} agent $A$ remains with $B$ at node $u$ at least $2\svI+ \svG{\svC}+\mathcal{H}$ rounds after round $t$. As a result, agent $A$ does not update $P$ in step $s$, which is a contradiction and proves the first step of the induction. Now consider there exits a positive integer $i'$ such that the property holds for all $i\leq i'$. If the last step of \sth{\swA} is step $s+i'$, the lemma directly follows. Otherwise, note that agent $A$ begins step $s+i'+1$ with the exact same map as in step $s+i'$. Hence using the same arguments as in step $s+i'$, agent $A$ does not update $P$ in step $s+i'+1$. This closes the induction and proves the lemma.   
\end{proof}

Note that in view of Lemmas~\ref{block:lem:snapshot} and~\ref{Block:lem:together}, we know that the imperfect map of every searcher of $\mathcal{E}$ remains always useful. In other terms, each of them always have a target node in every step of \sth{\swA}. This is stated in the following proposition.

\begin{proposition}
\label{prop:useful}
The map of every searcher of $\mathcal{E}$ is always useful.
\end{proposition}

In order to prove the main result of this section, i.e., Theorem~\ref{Block:theo:Principal-2}, we need the next three lemmas.

\begin{lemma}\label{block:lem:no-update}
If $f < n$, then there exists an integer $s$ in $\{1,\cdots,\mathcal{S}\}$ such that no searcher in set $\mathcal{E}$ updates its imperfect map $P$ during its execution of step $s$.
\end{lemma}

\begin{proof}
Assume for the sake of a contradiction that for each $s$ in $\{1,\cdots,\mathcal{S}\}$ there is at least one searcher of $\mathcal{E}$ that updates the output of its imperfect map $P$ during its execution of step $s$. According to \sth{\swA}, every searcher $A$ in $\mathcal{E}$ executes $IM(0,n)$ to compute $P$. When $A$ performs $IM(0,n)$, $A$ records all the followers it meets during the execution of \svF{\svC} in $IM(0,n)$. In particular, for each visited node $A$ can record at most $f$ Byzantine agents. This leads to at most $f$ ``wrong'' labels in each list of $P$. Since $f < n$, in view of Lemma~\ref{Block:lem:together}, each searcher of $\mathcal{E}$ performs at most $n. \svG{\svC}$ updates of $P$. Note that, two distinct searchers of $\mathcal{E}$ which start executing \sth{\swA} from the same node and at the same round act exactly in the same manner: in particular, they traverse the same edges synchronously, compute the same imperfect map and make the same updates at the same time. Hence, taking into account the maximum delay $\svI$, we know that the number of rounds in which we have a searcher of $\mathcal{E}$ making an update of its imperfect map is upper bounded by $\mathcal{U}= \svI\svC^2\svG{\svC}$. However, according to the algorithm $S=\mathcal{U}+1$. Hence, we get a contradiction, which proves the lemma.
\end{proof}

In view of Lemma \ref{block:lem:no-update}, we can define $s_{min}$ as being the first step for which there is no updates made by a searcher of $\mathcal{E}$. 
 
\begin{lemma}\label{block:lem:xagents}
If $f <n$, then there exist a round $\alpha$ and a node $v$ such that $x$ agents meet on node $v$ at round $\alpha$, and a searcher of $\mathcal{E}$ is in \sta{\szD} at round $\alpha$.
\end{lemma}

\begin{proof}
In order to prove the lemma, we first proceed by proving a series of $4$ claims. We start by introducing some notations that will facilitate the conduct of this proof.


Let $\mathcal{Q}$ be the set of nodes verifying the following condition: a node $u$ is in $\mathcal{Q}$ if $u$ is a target node of a searcher of $\mathcal{E}$ in step $s_{min}$. In view of Proposition~\ref{prop:useful}, $\mathcal{Q}\ne\emptyset$. Let $\mathcal{F}_{\mathcal{Q}}$ be the set of followers of $\mathcal{E}$ being on a node of $\mathcal{Q}$ at the beginning of their execution of step $s_{min}$.
Let $\rho$ be the last round in which a follower of $\mathcal{E}$ is in \sta{\swO}{} before entering either \sta{\swG} or \sta{\szA} (at round $\rho+1$) during its execution of step $s_{min}$.

\noindent \textbf{Claim 1} At round $\rho$, every searcher $A$ of $\mathcal{E}$ is on its target node. Moreover, $A$ remains on its target node for at least $\mathcal{H}$ rounds after round $\rho$.

\noindent \textbf{Proof of Claim 1}
According to \sth{\swA} and the maximal delay $\svI$, at round $\rho$ every searcher has spent in step $s_{min}$ at least $\svI+ 3\svG{\svC}$ rounds and at most $3\svI+ 3\svG{\svC}$ rounds. Moreover, in view of the definition of step $s_{min}$ and Proposition~\ref{prop:useful}, we know that every searcher remains in its target node at least  $2\svI+ \svG{\svC}+\mathcal{H}$ rounds while in \sta{\szD} of step $s_{min}$. However, before entering \sta{\szD} of step $s_{min}$, each searcher spends at least $\svI+ 2\svG{\svC}$ rounds and at most $\svI+ 3\svG{\svC}$ in step $s_{min}$. Hence the claim follows.

\noindent \textbf{Claim 2}
Let $B$ be a follower of $\mathcal{F}_{\mathcal{Q}}$. Agent $B$ remains idle in \sta{\szA} on its initial node from round  {$\rho+1$} to round $\rho+\mathcal{H}$. 

\noindent \textbf{Proof of Claim 2}
Let $u$ be the initial node of $B$ and $\rho'$ the last round in which it is in \sta{\swO} of step $s_{min}$. Since $u$ is a target node of a searcher $A$ of $\mathcal{E}$, agent $A$ reaches $u$ after having spent at least $\svI+2\svG{\svC}$ rounds and at most $\svI+3\svG{\svC}$ rounds in step $s_{min}$. Since $B$ waits $2\svI+ 3\svG{\svC}$ in \sta{\swO} at the beginning of step $s_{min}$, in view of the maximum delay between any pair of agents of $\mathcal{E}$, $A$ reaches node $u$ while $B$ is still in \sta{\swO}. Moreover, by definition of step $s_{min}$, $A$ remains on $u$ during $2\svI+\svG{\svC}+ \mathcal{H}$ rounds (in \sta{\szD}). Hence according to \sth{\swA}, at round $\rho'$ agent $B$ has shared its initial node with agent $A$ for at most {$2\svI+\svG{\svC}$ rounds} and it enters \sta{\szA} at round $\rho'+1$. So, after $\rho'$, agent $A$ stays idle with $B$ for at least $\mathcal{H}$ rounds. This means in particular that $B$ is in \sta{\szA} from round $\rho'+1$ to $\rho'+\mathcal{H}$.

Let $diff=\rho-\rho'$. Note that $0\leq diff\leq \svI$. According to the description of \sta{\szA}, from round $\rho'+\mathcal{H}$ to $\rho'+\mathcal{H}+diff$, agent $B$ leaves \sta{\szA} only if $A$ leaves $u$ at some round in $\{\rho'+\mathcal{H},\ldots,\rho'+\mathcal{H}+diff\}$. However, this is impossible according to Claim~1 and the fact that $\rho'\in\{\rho-\svI+1;\rho-\svI+2,\ldots,\rho\}$: indeed $\rho'+\mathcal{H}>\rho+1$ and $\rho'+\mathcal{H}+diff=\rho+\mathcal{H}$. Hence agent $B$ remains in \szA{} from round $\rho'+1$ to round $\rho'+\mathcal{H}+diff$, which proves the claim.

\noindent \textbf{Claim 3} Among the nodes of $\mathcal{Q}$, at least $|\mathcal{Q}|-1$ of them host a Byzantine agent in every round from round {$\rho+1$} to round $\rho+\mathcal{H}$.

\noindent \textbf{Proof of Claim 3} Let $\mathcal{I}$ be the time interval between round {$\rho+1$} and round $\rho+\mathcal{H}$. We show that during $\mathcal{I}$, at least $|\mathcal{Q}|-1$ nodes of $\mathcal{Q}$ host a Byzantine agent. Let $B$ be the first follower of $\mathcal{E}$ that starts the execution of \swA{}: if there are several agents satisfying the condition, we choose the one with the smallest label. Let us denote by $\Delta_B$ the round in which $B$ starts the execution of \swA{}. From Claim 1, during time interval $\mathcal{I}$, every searcher is on its target node. That is,  there are $|\mathcal{Q}|$ distinct target nodes for the searchers of $\mathcal{E}$. Hence, from \sth{\swA} and Lemmas~\ref{block:lem:snapshot} and~\ref{Block:lem:together}, it follows that on each node of $Q$ there is at least one agent $B'$ being (or pretending to be) a follower such that its label is at most $\ell_B$. According to the definition of $B$, and in particular its unicity, we know that at least $|\mathcal{Q}|-1$ target nodes host a Byzantine agent from round {$\rho+1$} to $\rho+\mathcal{H}$. Hence the claim holds.

Let $\mathcal{X}=\svI+ {4\svG{\svC}}+(\svC\svG{\svC})(\svI\svC+\svC)(2\svG{\svC}+\svI)$. Note that $\mathcal{H}=(n+1)\mathcal{X} {+3}$. We show the following claim.

\noindent \textbf{Claim 4} Let ${\rho+1} \leq \nu \leq  \rho+\mathcal{H}-\mathcal{X}$ be a round, if any, such that no good follower of $\mathcal{E}$ enters \sta{\swG} from round $\nu$ to round $\nu +\mathcal{X}-1$. At least $x$ agents meet at some round in $\{\nu+1,\ldots,\nu+\mathcal{X}\}$.


\noindent \textbf{Proof of Claim 4}

Let $\mathcal{F}_{\mathcal{Q}'}$ be the set of followers of $\mathcal{E}$ that do not belong to $\mathcal{F}_{\mathcal{Q}}$ and do not enter \sta{\swG} of step $s_{min}$ by round $v-1$.
Let $\mathcal{F}_{\mathcal{Q}''}$ be the set of followers of $\mathcal{E}$ that do not belong to $\mathcal{F}_{\mathcal{Q}}$ and enter \sta{\swG} of step $s_{min}$ by round $v-1$. Let $\mathcal{Q}'$ be the set of initial nodes of agents in $\mathcal{F}_{\mathcal{Q}'}$. Note that every good follower belongs to $\mathcal{F}_{\mathcal{Q}}\cup\mathcal{F}_{\mathcal{Q}'}\cup\mathcal{F}_{\mathcal{Q}''}$.

Let $B$ be a follower of $\mathcal{F}_{\mathcal{Q}''}$. We first show that the map of $B$, when it is computed, is always useful in step $s_{min}$ till round $\rho+\mathcal{H}$ included. Note that in view of Claim~1, it is enough to prove that agent $B$ starts and finishes the execution of $IM(1,n)$ in $\{\rho+1,\ldots,\rho+\mathcal{H}\}$. In view of the definition of $\rho$ and Corollary~\ref{block:cor:delay}, we know that $B$ enters \sta{\swG} at some round in $\{\rho-\svI+1,\ldots,\nu-1\}$. Moreover, when a follower enters this state, it first waits $\svI$ before executing $IM(1,n)$ that lasts $2\svG{\svC}$ rounds. Hence $B$ starts and finishes $IM(1,n)$ in $\{\rho+1,\ldots,\nu+\svI+2\svG{\svC}\}$. However, $\nu+\svI+2\svG{\svC}\leq \rho+\mathcal{H}$, which proves that the map of $B$, when it is computed, remains always useful in $\{\rho+1,\ldots,\rho+\mathcal{H}\}$.

As mentioned above, at round $\nu+ \svI+2\svG{\svC}$, every follower of $\mathcal{F}_{\mathcal{Q}''}$ has completed its execution of $IM(1,n)$. Observe that when a good follower $B$ transits to \sta{\swG} from \\ \sta{\szA} on a node $u$ at some given round $w$ between  round {$\rho+1$} and round $\rho+\mathcal{H}$, every good follower on $u$ also transits to \sta{\swG} from \sta{\szA} at round $w$: moreover, these good followers behave in a same synchronous manner i.e., they execute the same actions in each round between $w$ to round {$\rho+\mathcal{H}$}. That is, the total number of distinct maps of the agents of $\mathcal{F}_{\mathcal{Q}''}$ at round  $\nu+ \svI+2\svG{\svC}$ is at most $(\svI. n+n)$: there are at most $\svI\svC$ distinct maps of the good followers that transit to \sta{\swG} from either\sta{\swO} or \sta{\szA} before round $\rho+1$ and at most $n$ additional distinct maps of the good followers that transit from\\ \sta{\szA} to \sta{\swG} after round $\rho$.

Next, assume that there exists a round {$\alpha'$} such that $\nu+ \svI+2\svG{\svC}\leq {\alpha'} \leq {\nu+\mathcal{X}-2\svG{\svC}}$ and no good follower of $\mathcal{F}_{\mathcal{Q}''}$ updates its imperfect map from round {$\alpha'$} to round ${\alpha' +2\svG{\svC}}$. We show that in this case, $x$ agents meet on the same node at some round in {$\{\alpha',\ldots,\alpha'+2\svG{\svC}\}$}. {Let $B$, $P$ and $j$ be respectively a follower of $\mathcal{F}_{\mathcal{Q}''}$, the imperfect map of $B$ and its index from round {$\alpha'$} to round ${\alpha' +2\svG{\svC}}$}. The target node of $B$ is the node that is reached after performing the first $(j-1)$ edge traversals of \svF{\svC} from the initial node of $B$. Agent $B$ updates its imperfect map $P$ only if on its target node, there is no follower $B'$ such that $\ell_{B'}$ is the smallest label in $L_j$ of $P$. Since there are no updates from round {$\alpha'$} to round {$\alpha'+ 2\svG{\svC}$}, at round  {$\alpha'+2\svG{\svC}$}, every follower $B$ of $\mathcal{F}_{\mathcal{Q}''}$ is on its target node $u$.

Let us consider the case where $u$ is neither in $\mathcal{Q}$ nor $\mathcal{Q}'$, we show that $u$ hosts at least one Byzantine agent. From \sth{\swA}, we know that at round {$\alpha'+2\svG{\svC}$}, node $u$ hosts a follower $B'$ such that  $\ell_{B'}$ is the smallest label in {$L_j$} of $P$. If $B'$ is a good follower, $B'$ is in \sta{\szA} with a searcher $A$ (recall that no good follower transits to \sta{\swG} from round $\nu$ to round {$\nu +\mathcal{X}-1$}). However, $A$ cannot be a good searcher of $\mathcal{E}$ since $u$ is not in $\mathcal{Q}$. Hence, $u$ hosts indeed a Byzantine agent at round {$\alpha'+2\svG{\svC}$}. Note that in view of the definition of $\nu$ and the algorithm, each agent of $\mathcal{F}_{\mathcal{Q}'}$ is on its initial node with a Byzantine agent pretending to be a searcher from round $\rho+1$ to $\nu+\mathcal{X}-1$ (as all the good searchers are in nodes $\notin\mathcal{Q}'$ according to Claim~1). Let $\mathcal{Q}''$ be the target nodes which do not belong to $\mathcal{Q}\cup\mathcal{Q}'$, of the good followers of $\mathcal{F}_{\mathcal{Q}''}$ at round   {$\alpha'+2\svG{\svC}$}. By Claim 3, we then have $|\mathcal{Q}|+|\mathcal{Q}'|+ |\mathcal{Q}''| \leq f+1$. Moreover, at round  {$\alpha'+2 \svG{\svC}$}, every good agent is in a node of $\mathcal{Q}\cup\mathcal{Q}'\cup\mathcal{Q}''$. Hence by the Pigeonhole principle, it follows that $x$ agents share the same node at round  {$\alpha'+2 \svG{\svC}$}. If round  {$\alpha'$} exists, then the claim holds. So to conclude the proof of this claim, it remains to show the existence of round  {$\alpha'$}. Recall that each follower of $\mathcal{F}_{\mathcal{Q}''}$ performs at most $\svG{\svC}. n$ updates of its imperfect map $\mathcal{P}$ (since it can record at most $f$ Byzantine agents that pretend to be followers in \sta{\szA} on each node during the execution of $IM(1,n)$). Besides, as argued earlier, the total number of distinct maps of the agents of $\mathcal{F}_{\mathcal{Q}''}$ at round  $\nu+ \svI+2\svG{\svC}$ is at most $(\svI. n+n)$. So, after at most $(\svI. n+n).(\svG{\svC}. n)(\svI+2 \svG{\svC}$ $=\mathcal{X}-\svI-4\svG{\svC}$ rounds from {$\nu+\svI+2\svG{\svC}$}, no good follower of $\mathcal{F}_{\mathcal{Q}''}$ updates its imperfect map. Moreover, every good follower of $\mathcal{F}_{\mathcal{Q}''}$ spends at most  {$2\svG{\svC}$} rounds before reaching its target node. This proves the existence of round {$\alpha'$} and by extension the claim.
  

We are now able to prove our lemma. Assume by contradiction that the lemma does not hold. This means either there is no round when $x$ agents meet, or in every round $z$ when $x$ agent meet, no searcher of $\mathcal{E}$ is in \sta{\szD} at round $z$.
Let us first consider the former case. Let $\mathcal{F}'$ be the set of good followers that enter \sta{\swG} from \sta{\szA} at some round in $\{\rho+1,\ldots,\rho+\mathcal{H}\}$. From Claim 4, we can deduce that there is no consecutive $\mathcal{X}$ rounds in  {$\{\rho+1,\ldots,\rho+\mathcal{H}-\mathcal{X}\}$} in which no good follower of $\mathcal{E}$ transits to \sta{\swG} (otherwise, round $\alpha$, which is defined in the statement of this lemma, exits). From round $\rho+2$ to $\rho+\mathcal{H}$, only the followers of $\mathcal{F}'$ may enter \sta{\swG}. From round $\rho+2$ all the agents of $\mathcal{F}'$ have already entered \sta{\szA} in view of the definition of $\rho$.
Note that $|\mathcal{Q}_{\mathcal{F}'}| \leq n$ where $|\mathcal{Q}_{\mathcal{F}'}|$ is the set of initial nodes of at least one follower of $\mathcal{F}'$. Moreover, let $C$ be an agent of $\mathcal{F}'$ that enters \sta{\swG} from \sta{\szA} at a round $t\in\{\rho+2,\ldots,\rho+\mathcal{H}\}$: before round $t$, agent $C$ does not move in step $s_{min}$, and all the agents of $\mathcal{F}'$ that are in \sta{\szA} and share the same node as $C$ in round $t-1$ also enter \sta{\swG} at round $t$.
Hence, after at most $n\mathcal{X}$ rounds from round $\rho+2$, there is no agent that can enter \sta{\swG} till round $\rho+\mathcal{H}$ included. However round $\rho+3+n\mathcal{X}\leq\rho+\mathcal{H}-\mathcal{X}$. Hence there exists a round $v$ satisfying the statement of Claim~4 and there is a meeting of at least $x$ agents at some round in $\{\nu+1,\ldots,\nu+\mathcal{X}\}$: we get a contradiction with the fact that $\alpha$ does not exist.
Concerning the latter case, note that there is a round $\alpha$ in $\{\nu+1,\ldots,\nu+\mathcal{X}\}$ in which $x$ agents meet. In view of Claim 1 and \sth{\swA}, every searcher of $\mathcal{E}$ is in \sta{\szD} in every round belonging to $\{\nu+1,\ldots,\nu+\mathcal{X}\}$ : we get a contradiction with the fact that no searcher of $\mathcal{E}$ is in \sta{\szD} at round $\alpha$.


\end{proof}

{{\begin{lemma}\label{Block:lem:groupx-f}
	If there exists a round $r$ at which at least $x \geq f + 2$ agents meet on the same node and among them all the good ones are executing phase Process at round $r$, then at least $(x-f)$ good agents exit their execution of \swA{} at the same round and on the same node. 
\end{lemma}

\begin{proof}
	Assume there exists such a round. Let us show that $(x-f)$ good agents exit their execution of \swA{} at the same round and on the same node. Let $x'$ be the largest number of agents executing \swA{} but not in \sta{\szE} which met in the same node $u$ in some round $\Delta+ w$. If there are several such rounds, we consider the one with the largest value of $w$. The good agents executing \swA{} but in another state than \szE{} are precisely those executing phase Process, which implies that $x' \geq x$. Let $\mathcal{Y}$ be the set of good agents executing phase Process on $u$ at round $\Delta+w$. Remark that at round $\Delta + w$ on $u$ there are at most $f$ Byzantine agents. Hence, $|\mathcal{Y}|\geq x'-f$.
	
	When in \sta{\szE}, every agent $A$ of $\mathcal{Y}$ repeats exactly the same waiting periods and edge traversals as in its execution phase Process in order to reconstruct the group of agents that was at node $u$ in round $\Delta+ w$. More precisely, let $r$ and $r'$ be the round when $A$ initiated \swA{} and the round when $A$ enters \sta{\szE} respectively. Let $i$ be an integer such that $r+i=\Delta+w$. From round $r'$ to $r'+i-1$, agent $A$ replays exactly the same waiting periods and edges traversals from round $r$ to $r+i-1$: for each integer $y$ in $\{0,1,\ldots,i-1\}$, if agent $A$ remains idle (resp. leaves the current node via a port $o$) from round $r+y$ to round $r+y+1$, then agent $A$ remains idle (resp. leaves the current node via port $o$) from round $r'+y$ to $r'+y+1$. In round $r'+i$, agent $A$ is in node $u$  and stops the execution of \swA{}. Besides, in view of Lemma~\ref{lem:duration}, every good agent spends the same number of rounds executing phase Process: let us denote this number by $\mathcal{W}$. So, $r'+i=r+\mathcal{W}+i=\Delta+w+\mathcal{W}$. Hence, every agent of $\mathcal{Y}$ is in node $u$ and stops the execution of \swA{} at round $\Delta+w+\mathcal{W}$.
\end{proof}}}
 


Now we are ready to end this subsection by giving the main theorem related to \sth{\swA}. In order to use the theorem outside of this subsection, we recall in the statement the assumptions that were made in the beginning of this subsection.

\begin{theorem}\label{Block:theo:Principal-2}
Consider a team made of at least $(x-1)(f+1)+1$ good agents in a graph of size at most $n$, where $x\geq f+2$. Let $\Delta$ be the first round when a good agent starts executing \swAA{\svI}{\svC}{\svX}. If all good agents start executing \swAA{\svI}{\svC}{\svX} by round $\Delta+\svI-1$, and parameter $bin$ is $0$ (resp. $1$) for at least one good agent, then we have the following property. After at most a time polynomial in $n$ and $\svI$ from $\Delta$, at least $(x-f)$ good agents finish the execution of \swA{} at the same round and in the same node.
\end{theorem}

{{\begin{proof}
	When in \sta{\szE}, an agent only replays all or part of the waiting periods and edge traversals made in phase Process. Hence, according to Lemma~\ref{lem:duration} and the initial delay that is at most $\svI$, we know that every good agent finishes the execution of \swA{} after at most a time polynomial in $n$ and $\svI$ from $\Delta$.
	
	So to prove the theorem it remains just to show that there is a group of at least $(x-f)$ good agents that exit \swA{} on the same node and at the same time. This follows directly from Lemma~\ref{Block:lem:groupx-f} and the claim that is proven below.

\noindent \textbf{Claim~1} At least $x$ agents meet on the same node at some round $t$, and among them all the good ones are executing phase Process of \sth{\swA} at round $t$.

\noindent \textbf{Proof of Claim~1}

If $f \geq n$, there are always $x$ agents sharing the same node as the number of good agent is at least $(f+1)x$. Moreover, at round $\Delta+\svI$ every good agent is executing phase process of \sth{\swA}. Hence, the claim holds if $f \geq n$. 

So let us focus on the case where $f < n$. From Lemma~\ref{block:lem:xagents}, there is a round $\alpha$ when $x$ agents meet in some node $v$ and there is a good searcher $A$ in \sta{\szD} of some step $s$ in round $\alpha$. At round $\alpha$, it remains for agent $A$ at least $\svI$ rounds to spend in step $s$. Indeed, in \sta{\szB} of step $s$, an agent has to wait $5\svI+5\svG{\svC}+\mathcal{H}-x$ rounds and $x$ is upperbounded by $4\svI+5\svG{\svC}+\mathcal{H}$ (this is shown in the proof of Lemma~\ref{lem:duration}). Hence, in view of Corollary~\ref{block:cor:delay}, no good agent has finished step $s$ of phase Process at round $\alpha$. Moreover, agent $A$ has necessarily spent more than $\svI$ rounds in step $s$ when in round $\alpha$. So, every good agent is executing phase Process of \sth{\swA} at round $\alpha$, which proves the claim.
\end{proof}}}

\subsection{\stH{\swD}}

	The second building block called \swD{} takes as input two integers $n$ and $\svI$. \stJ{\swDD{\svI}{\svC}} allows all the good agents to finish their executions of the subroutine in the same node and at the same round, provided the following two conditions are satisfied. The first condition is that all good agents are in a graph of size at most $n$ and start executing \swDD{\svI}{\svC} in an interval lasting at most $\svI$ rounds. The second condition is that at least $4\svA + 2$ good agents start executing \swDD{\svI}{\svC} at the same round and in the same node. The time complexity of the procedure is polynomial in $\svI$ and $n$.
	
	\subsubsection{High level idea}
	

	For the sake of convenience, we will consider in this subsubsection that a group of agents is a set of all agents, at least one of which is good, that start executing \sth{\swD} in the same node and at the same round. In the sequel, we assume there is a group of at least $4f+2$ good agents. The reasons why we need such an assumption will appear at the end of the explanations. Let $G_{max}$ and $v_{max}$ be respectively the group with the largest initial number of agents and its starting node. In case there are several possible groups $G_{max}$, we choose among them the one having the largest lexicographically ordered list of pairwise distinct labels denoted by $L_{max}$:  this guarantees the unicity of $G_{max}$ as it contains at least $4f+2$ good agents. The cardinality of a list $L$ will be denoted by $|L|$.
		
	The idea underlying \sth{\swD} is to make all good agents elect the same node, and then gather in it (if we ensure this, then we can ensure that all good agents finish the execution of \swD{} {at a same round} using some technicalities). Each node is a candidate, and each good agent supports the node in which it started executing the procedure. Besides supporting its candidate, each good agent is also a voter. When acting as a supporter, a good agent stays idle to promote its candidate and when acting as a voter, it makes a traversal of the graph in order to visit all nodes of the graph (using procedure \svF{\svC}), and then elects one of the nodes using the information provided by the supporters. In order to establish such a strategy, note that all good agents must not act as voters at the same time. Otherwise, there would be no supporter left in its candidate node to promote it. Hence, the election process is divided into two parts, and each group is divided into two subgroups of nearly equal size using the labels of the agents. During the first (resp. second) part of the election, the first (resp. second) subgroup of each {group} acts as voters while the second (resp. first) subgroup of each {group} acts as supporters.



When visiting a node during its traversal of the graph, a voter gets from each supporter of this node a promotional information: for a good supporter, it is simply the lexicographically ordered list of all pairwise distinct labels of the agents that were initially in its group. Once its traversal is done, the voter considers each node $v$ satisfying the property that at least $\svV{\svP{L}}$ distinct agents in $v$ have transmitted a lexicographically ordered list $L$. Among these nodes, the voter elects the one for which the property is true with the list $L$ having the largest cardinality: in case of a tie, the lexicographical order on the labels is used as done to ensure the unicity of $G_{max}$. By doing so, all good agents elect node $v_{max}$ and then gather in it: the purpose of the last paragraph is to explain why we have the guarantee that $v_{max}$ is unanimously elected.

	By definition, the number of good agents that is initially in $G_{max}$, and thus $|L_{max}|$ is at least $4f+2$. Moreover, the number of Byzantine agents is initially at most $f$ in $G_{max}$. Hence, we can show that our strategy permits to always have at least $\svV{\svP{L_{max}}}$ distinct agents in $v_{max}$ that transmit list $L_{max}$ to all voters. Note that each good supporter transmits a list $L$ such that $|L|<|L_{max}|$, or $|L|=|L_{max}|$ and $L$ is not lexicographically larger than $L_{max}$. So, the only way the Byzantine agents could prevent the good agents to elect $v_{max}$ would be that at least $\svV{\svP{L'}}$ Byzantine agents transmit a list $L'$ such that $|L'|>|L_{max}|$, or $|L'|=|L_{max}|$ and $L'$ is lexicographically larger than $L_{max}$. However this situation is impossible because the Byzantine agents are not numerous enough: indeed $\svV{\svP{L'}}\geq f+1$.

	\subsubsection{Formal description of the algorithm}
	
		When an agent $A$ executes \swDD{\svI}{\svC}, it can transit to different states that are \swQ{}, \swR{} and \swZ{}. When agent $\sxa$ starts the execution of \swD{}, it is in \sta{\swQ}. In the algorithm, the cardinality of a list $L$ will be denoted by $|L|$.

		\stb{\swQ} Agent $\sxa$ spends a single round in this state. Besides its state, it transmits its label to the agents sharing the same node. Agent $\sxa$ assigns to variable $\svg$, the lexicographically ordered list of all pairwise distinct labels of agents that are currently in its node and in \sta{\swQ}. Then $\sxa$ transits to \sta{\swR}.
		
		\stb{\swR} When it enters this state, agent $\sxa$ initializes two variables: it assigns an empty list to variable $\svh$, and 0 to variable $\svi$. This state is made of five different periods: the first, third and fifth (resp. the second and fourth) ones are waiting periods (resp. moving periods). In each round of the two first waiting periods, agent $\sxa$ transmits the list $\svg$ built when in \sta{\swQ}. If $\sxb$ belongs to the first $\svQ{\frac {\svP{\svg}} {2}}$ labels of $\svg$, then the durations of the two first waiting periods are respectively $\svI - 1$ and $\svI + 2\svG{\svC} - 1$. Otherwise, they respectively last $\svI + 2\svG{\svC} - 1$ and $\svI - 1$ rounds. The duration of the third waiting period is given after describing the second moving period. 

During the first moving period, agent $\sxa$ executes \svF{\svC} followed by a backtrack in which the agent traverses all edges traversed in \svF{\svC} in the reverse order. Once this backtrack is done, the agent assigns to variable $I$ the largest list $\sxg$, if any, having the following property: there is a round during the execution of \svF{\svC} at which agent $A$ is in a node where at least $\svV{\svP{\sxg}}$ distinct agents in \sta{\swR} transmit $\sxg$. (We consider that a list $\sxj$ is larger than another list $\sxk$ if and only if $\sxj$ contains more elements, or $\sxj$ and $\sxk$ contain the same number of elements and $\sxj$ is lexicographically larger than $\sxk$). If such a list $\sxg$ exists, the agent also assigns to variable $\svi$, the smallest number of edge traversals made by $\sxa$ during the execution of \svF{\svC} to reach a node satisfying the above property with $\sxg$. Otherwise, the agent leaves variables $\svh$ and $\svi$ unchanged.

During the second moving period, agent $\sxa$ performs the first $\svi$ edge traversals of \svF{\svC}. Once this is done, agent $\sxa$ checks whether $\svg = \svh$ or not. If $\svg = \svh$, then the third waiting period lasts $\svI + \svG{\svC} - 1$ rounds, and at its expiration, $\sxa$ transits to \sta{\swZ}. Otherwise, the third waiting period lasts $2\svI + \svG{\svC} - 1$ but can be interrupted when agent $\sxa$ notices at least $\svV{3\svP{\svh}}$ agents in \sta{\swZ} in its node: as soon as such an event occurs, agent $A$ exits the execution of \swDD{\svI}{\svC}. In case such an interruption does not occur, the agent exits the execution of \swDD{\svI}{\svC} at the end of the waiting period.

		\stb{\swZ} Agent $\sxa$ spends one round in this state and then exits the execution of \swDD{\svI}{\svC}.
		
		
	\subsubsection{Correctness and complexity analysis}
	
Concerning \sth{\swD}, we only have the following theorem.

		\begin{theorem} \label{block2:proof:theo}
			Consider a team of agents in a graph of size at most $n$. Let $\sxe$ be the first round when a good agent starts executing \swDD{\svI}{\svC}. If every good agent starts executing \swDD{\svI}{\svC} by round $\sxe + \svI -1$ and among them at least $4\svA + 2$ start the execution in the same node and at the same round, then all good agents finish their executions of \sth{\swD} in the same node and at the same round $r< \sxe+4\svI + 6\svG{\svC} - 1$.
		\end{theorem}
			
		\begin{proof}
			Note that according to procedure \swD{}, every good agent spends at most $4\svI + 6\svG{\svC} - 1$ rounds in any execution of \sth{\swDD{\svI}{\svC}}. Hence, to prove the theorem we just have to prove that all good agents finish their executions of \sth{\swD} in the same node and at the same round.	
			
			Let us denote by $\sxz$ the largest list $\svg$ built by any good agent in \sta{\swQ}, and by $\sya$ one of the good agents that builds it. By assumption, they are at least $4\svA + 2$ good agents that start the execution in the same node and at the same round. As a result, in view of the description of \sta{\swQ}, $\sxz$ contains at least $4\svA + 2$ elements, and agent $\sya$ belongs to the group of at least $3\svA + 2$ good agents in \sta{\swQ} that compute the same list $\sxz$ at a round $\syb$ in a node $\syc$. Let us call $\syd$ the group of all the good agents in \sta{\swQ} in node $\syc$ at round $\syb$. We prove the following two claims.

\noindent \textbf{Claim 1} The agents of $\syd$ are the only good agents that build list $\sxz$ while in \sta{\swQ}.

\noindent \textbf{Proof of Claim 1} Let us assume by contradiction that the claim is false. Hence, there is a good agent $\sye$ in \sta{\swQ} which also builds $\sxz$ in a node $\syf$ at a round $\syg$ such that $\syf\neq\syc$ or $\syg\neq\syb$. In view of the description of \sta{\swQ}, there are all the labels of the agents of $\syd$ in $\sxz$. Thus, for each good agent of $\syd$, there is an agent in \sta{\swQ} with the same label in node $\syf$ at round $\syg$. However, there are at least $3\svA + 2$ agents in $\syd$, and since they only spend round $\syb$ in \sta{\swQ} in node $\syc$, none of them is in this state in node $\syf$ at round $\syg$. Besides, all the good agents have different labels and the Byzantine agents are not numerous enough to be these $3\svA + 2$ agents in \sta{\swQ} in node $\syf$ at round $\syg$. This contradicts the existence of these $3\svA + 2$ agents and the assumption that $\sye$ builds $\sxz$ in node $\syf$ at round $\syg$. Hence, the claim is proven.

			
\noindent \textbf{Claim 2} Each good agent starts its third waiting period in node $\syc$

\noindent \textbf{Proof of Claim 2}

To prove the claim, we first need to prove the following two facts. The first fact is that in each of the rounds belonging to $\{\syb + 1,\ldots,\syb + 2\svI + 4\svG{\svC} - 2\}$, there are at least $\svV{\svP{\sxz}}$ good agents in \sta{\swR} that transmit the list $\sxz$ in node $\syc$. The second fact is that each good agent performs entirely its first moving period between round $\syb + 1$ and round $\syb + 2\svI + 4\svG{\svC} - 2$.

			
Let us focus on the first fact. In view of the description of \sta{\swQ}, the list $\sxz$ contains at least $3\svA + 2$ elements corresponding to the labels of the agents of $\syd$, all of which are good, and at most $\svP{\syd} + \svA$ elements, with $\svP{\syd}$ the number of agents in $\syd$. This means that $\svA < \svV{\svP{\sxz}}$, $\svP{\syd} > \svV{3\svP{\sxz}}$, $\svQ{\frac {\svP{\sxz}} {2}} - \svA \geq \svV{\svP{\sxz}}$ and $\svZ{\frac {\svP{\sxz}} {2}} - \svA \geq \svV{\svP{\sxz}}$ \ie in each half of $\sxz$ there are at least $\svV{\svP{\sxz}}$ labels of agents of $\syd$. This implies that in each of the rounds belonging to $\{\syb + \svI,\ldots,\syb + \svI + 4\svG{\svC} - 1\}$, there are at least $\svV{\svP{\sxz}}$ good agents in \sta{\swR} transmitting the list $\sxz$ in node $\syc$. Moreover, in view of the description of \sta{\swR}, all the agents of $\syd$ wait in $\syc$ and transmit $\sxz$ in each round from round $\syb + 1$ to round $\syb + \svI - 1$, and from round $\syb + \svI + 4\svG{\svC}$ to round $\syb + 2\svI + 4\svG{\svC} - 2$. Hence, the first fact is true.

Let us go further by considering the second fact. Each good agent starts the execution of \sth{\swD} between rounds $\sxe$ and $\sxe + \svI - 1$. Then, it spends a single round in \sta{\swQ}, and enters \sta{\swR} between round $\sxe + 1$ and round $\sxe + \svI$. Actually, the good agents of $\syd$ are in \sta{\swQ} at round $\syb$. This means that $\syb$ belongs to $\{\sxe; \ldots; \sxe + \svI - 1\}$. Since every good agent spends at least $\svI - 1$ rounds and at most $\svI + 2\svG{\svC} - 1$ rounds in the first waiting period, every good agent starts its first moving period between round $\sxe + \svI$ and round $\sxe + 2\svI + 2\svG{\svC} - 1$ \ie between round $\syb + 1$ and round $\syb + 2\svI + 2\svG{\svC} - 1$. Since the first moving period lasts $2\svG{\svC}$ rounds, the second {fact} is true.

Hence, from the two facts, we know that during its first moving period each good agent visits $\syc$ and notices at least $\svV{\svP{\sxz}}$ agents in \sta{\swR} transmitting the same list $\sxz$. As a result, in view of the description of \sta{\swR} each good agent finishes the second moving period at round $\syc$ except if the following event occurs: there is a list $\syi$ strictly larger than or identical to $\sxz$ such that at a round $\syj$, in a node $\syk \ne \syc$, at least $\svV{\svP{\syi}}$ agents in \sta{\swR} transmit $\syi$ to a good agent while it is performing the \svF{\svC} of its first moving period. However, such an event cannot occur. Let us assume by contradiction it can. Since $\svP{\syi} \geq \svP{\sxz} \geq 4\svA + 2$, among $\svV{\svP{\syi}} > \svA$ agents in \sta{\swR} transmitting $\syi$, there must be at least one good agent which builds $\syi$ in \sta{\swQ}. Note that either $\syi$ is identical to $\sxz$ or it is larger than $\sxz$. If $\syi$ is identical to $\sxz$ we get a contradiction with Claim~1. If $\syi$ is larger than $\sxz$, we get a contradiction with the maximality of $\sxz$, which concludes the proof of the claim.

	In view of Claim~2 and the description of states~\swQ{} and~\swR{}, we know that every good agent finishes its execution in the same node. Hence, to conclude the proof of the theorem, we just have to prove now that all good agents finish the execution at the same time. To do this, in view of the fact that $\svP{\syd} \geq \svV{3\svP{\sxz}}$ and the fact that each good agent assigns to variable $\svh$ the same list $\sxz$ at the end of its first moving period, it is enough to show that there is a round in which the good agents of $\syd$ are in \sta{\swZ} and all the others good agents are performing their third waiting period. It is the purpose of the following lines.

First assume that no good agent prematurely interrupts its third waiting period before round $\syb + 3\svI + 5\svG{\svC} - 2$. Since each good agent assigns to variable $\svh$ the same list $\sxz$, each agent of $\syd$ performs no edge traversal in the second moving period and enters \sta{\swZ} at round $\syb + 3\svI + 5\svG{\svC} - 2$. Each good agent starts its first waiting period between round $\syb - \svI + 2$ and round $\syb + \svI$. Moreover, it can spend from 0 to $\svG{\svC}$ rounds in its second moving period. This implies that each good agent completes it between round $\syb + \svI + 4\svG{\svC} - 1$ and round $\syb + 3\svI + 5\svG{\svC} - 3$ and starts the third waiting period between round $\syb + \svI + 4\svG{\svC}$ and round $\syb + 3\svI + 5\svG{\svC} - 2$. Furthermore, each good agent that does not belong to $\syd$ assigns to variable $\svh$ a list that is different from the list it has built when in \sta{\swQ}, and thus its third waiting period lasts $2\svI + \svG{\svC} - 1$ rounds. This means that each good agent which does not belong to $\syd$ completes its third waiting period between round $\syb + 3\svI + 5\svG{\svC} - 2$ and round $\syb + 5\svI + 6\svG{\svC} - 4$. Hence, each good agent that does not belong to $\syd$ is performing its third waiting period at round $\syb + 3\svI + 5\svG{\svC} - 2$ when the agents of $\syd$ enter \sta{\swZ}. As a result, the theorem is true if no good agent prematurely interrupts its third waiting period before round $\syb + 3\svI + 5\svG{\svC} - 2$. However, no good agent can interrupt its third waiting period at a round $r<\syb + 3\svI + 5\svG{\svC} - 2$. Indeed, if it was the case, that would {imply} that there are at least $\svV{3\svP{\sxz}}$ agents in \sta{\swZ} at round $r$ and among them there is {necessarily} one good agent of $\syd$: this {contradicts} the fact that the agents of $\syd$ {enter} \sta{\swZ} at round $\syb + 3\svI + 5\svG{\svC} - 2$. This ends the proof of the theorem.
		\end{proof}

%% file: Sections/onlynknown.tex
\section{The positive result} \label{sec:onlyn}

			In this section we show an algorithm, called \swF{}, that solves $\svA$-Byzantine gathering with strong teams in all graphs of size at most $\svC$, assuming that $\svH=\svN$: note that such a global knowledge can be coded using $\std{\log \log \log \svC}$ bits. The algorithm works in a time polynomial in $\svC$ and $|\svB|$, and it makes use of the building blocks introduced in the previous section. 

In the sequel, we denote by $\swAa{\svC}$ the maximal time complexity of \sth{\swAA{\svG{\svC}}{\svC}{\svf}} with $\svf\in\{0;1\}$ in all graphs of size at most $\svC$. We also denote by $\swDd{\svC}$ the maximal time complexity of \sth{\swDD{\svG{\svC}+\swAa{\svC}}{\svC}} in all graphs of size at most $\svC$. Note that according to Theorems~\ref{Block:theo:Principal-2} and~\ref{block2:proof:theo}, $\swAa{\svC}$ and $\swDd{\svC}$ exist and are polynomials in $\svC$.
		
		
		
		
		\subsection{Intuition}

In order to better describe the high level idea of our solution, let us first consider a situation that would be ideal to solve Byzantine gathering with a strong team and that would be as follows. Instead of assigning distinct labels to all agents, the adversary assigns to each of them just one bit $\rho\in\{0;1\}$, so that there are at least one good agent for which $\rho=0$ and at least one good agent for which $\rho=1$. Such a situation would clearly constitute an infringement of our model, but would allow the simple protocol described in Algorithm~\ref{onlyn:simple} to solve the problem in a time that is polynomial in $n$ when $\svH=\svN$. Let us briefly explain why.

			\begin{algorithm}
				\caption{{Algorithm executed by every good agent in the ideal situation.}}
				\label{onlyn:simple}
				\begin{footnotesize}
				\begin{algorithmic}[1]
					\State Let $\rho$ be the bit assigned to me by the adversary
					\State Execute $\mathcal{A}(\rho)$
					\State Declare that gathering is achieved
				\end{algorithmic}
				\end{footnotesize}
			\end{algorithm}

\begin{algorithm}
				\caption{$\mathcal{A}(\rho)$ executed by a good agent.}
				\label{onlyn:small}
				\begin{footnotesize}
				\begin{algorithmic}[1]
					\State $N \gets 2^{(2^{\svH})}$\label{l:1}
					\State Execute \svF{\svD} 
					\State Execute \swAA{\svG{\svD}}{\svD}{\svf}
					\State Execute \swDD{\svG{\svD} + \swAa{\svD}}{\svD}
				\end{algorithmic}
				\end{footnotesize}
			\end{algorithm}

Algorithm~\ref{onlyn:simple} consists mainly of a call to $\mathcal{A}(\rho)$ that is given by Algorithm~\ref{onlyn:small}.
Since $\svH=\svN$, we know that at line~\ref{l:1} of Algorithm~\ref{onlyn:small}, $N$ is a polynomial upperbound on $n$, and the execution of \svF{\svD} in a call to $\mathcal{A}(\rho)$ by the first woken-up good agent permits to visit every node of the graph and to wake up all dormant agents. As a result, the delay between the starting times of \swAA{\svG{\svD}}{\svD}{\svf} by any two good agents of the strong team is at most $\svG{\svD}$. According to the properties of procedure \swA ~(cf. Theorem~\ref{Block:theo:Principal-2}), this guarantees in turn that the delay between the starting times of \swDD{\svG{\svD} + \swAa{\svD}}{\svD} by any two good agents is at most $\svG{\svD} + \swAa{\svD}$, and at least $4f+2$ good agents start this procedure at the same time in the same node. Hence, in view of the properties of procedure \swD ~(cf. Theorem~\ref{block2:proof:theo}), all good agents declare gathering is achieved at the same time in the same node after a polynomial number of rounds (w.r.t $n$) since the wake-up time of the earliest good agent.

Unfortunately, we are not in such an ideal situation. At first glance, one might argue that it is not really a problem because all agents are assigned distinct labels that are, after all, distinct binary strings. Thus, by ensuring that each good agent applies on its label the transformation given in Section~\ref{sec:pre}, and then processes one by one each bit $b_i$ of its transformed label by executing $\mathcal{A}(b_i)$, we can guarantee (with some minor technical adjustments) that the gathering of all good agents is done in time polynomial in $n$ and $|l_{min}|$. Indeed, in view of Proposition~\ref{prelim:label} the conditions of the ideal situation are recreated when the agents process their $j$-th bits for some $j\leq2|l_{min}|+4$. Unfortunately it is not enough for our purpose. In fact, in the ideal situation, there is just one bit to process: thus, {\em de facto} every good agent knows that every good agent knows that gathering will be done at the end of this single process. However, it is no longer the case when the agents have to deal with sequences of bit processes: the good agents have a priori no mean to detect collectively and simultaneously when they are gathered. It should be noted that if the agents knew $f$, we could use an existing algorithmic component (cf. \cite{DieudonnePP14}) allowing to solve $f$-Byzantine gathering if at some point some good agents detect the presence of a group of at least $2f+1$ agents in the network. Such a group is necessarily constructed during the sequence of bit processes given above, but again, it cannot be a priori detected as the agents {do} not know $f$ or an upperbound on it. Hence, in our goal to optimize the amount of global knowledge, we need to implement a new strategy to allow the good agents to declare gathering achieved jointly and simultaneously. It is the purpose of the rest of this subsection.

To get all good agents declare simultaneously the gathering achieved, we want to reach a round in which every good agent knows that every good agent knows that gathering is done. So, let us return to our sequence of bit processes. As mentioned above, when a good agent has finished to read the first half of its transformed label -- call such an agent \emph{experienced} -- it has the guarantee that the gathering of all good agents has been done at least once. Hence, when an experienced agent starts to process the second half of its transformed label, it actually knows an approximation of the number of good agents with a margin of error of $f$ at the most. For the sake of convenience, let us consider that an experienced agent knows the exact number $\mu$ of good agents: the general case adds a slight level of complexity that is unnecessary to understand the intuition. So, each time an experienced agent completes the process of a bit in the second half of its transformed label, it is in a node containing less than $\mu$ agents or at least $\mu$ agents. In the first case, the experienced agent is sure that the gathering is not achieved. In the second case, the experienced agent is in doubt. In our solution, we build on this doubt. How do we do that? So far, each bit process was just made of one call to procedure $\mathcal{A}$: now at the end of each bit process, we add a waiting period of some prescribed length, followed by an extra step that consists in applying $\mathcal{A}$ again, but this time according to the following rule. If during the waiting period it has just done, an agent $X$ was in a node containing, for a sufficiently long period, an agent pretending to be experienced and in doubt (this agent may be $X$ itself), then agent $X$ is said to be {\em optimistic} and the second step corresponds to the execution of $\mathcal{A}(0)$. Otherwise, agent $X$ is said to be {\em pessimistic} and the second step corresponds to the execution of $\mathcal{A}(1)$.


If at least one good agent is optimistic within a given second step, then the gathering of all good agents is done at the end of this step. Indeed, through similar arguments of partition to those used for the ideal situation, we can show it is the case when at least another agent is pessimistic. However, it is also, more curiously, the case when there is no pessimistic agents at all. This is due in part to the fact that two good experienced agents cannot have been {in doubt in two distinct nodes} during the previous waiting period (otherwise, we would get a contradiction with the definition of $\mu$). Thus, all good agents start $\mathcal{A}(0)$ from at most $f+1$ distinct nodes (as the Byzantine agents can mislead the good agents in at most $f$ distinct nodes during the waiting period), which implies by the pigeonhole principle that at least $4f+2$ good agents {start} it from the same node. Combined with some other technical arguments, we can show that the conditions of Theorem~\ref{block2:proof:theo} are fulfilled when the agents execute \swD{} at the end of $\mathcal{A}(0)$, thereby guaranteeing again gathering of all good agents. 

As a result, the addition of an extra step to each bit process gives us the following interesting property: when a good agent is optimistic at the beginning of a second step, {at its end} the gathering is done and, more importantly, the optimistic agent knows it because its existence ensures it. Note that, it is a great progress, but unfortunately it is not yet sufficient, particularly because the pessimistic agents do not have the same kind of guarantee. The way of remedying this is to repeat once more the same kind of algorithmic ingredient as above. More precisely, at the end of each second step, we add again a waiting period of some prescribed length, followed by a third step that consists in applying $\mathcal{A}$ in the following manner. If during the waiting period it has just done, an agent $X$ was in a node containing, for a sufficiently long period, an agent pretending to be optimistic, then the third step of agent $X$ corresponds to the execution of $\mathcal{A}(0)$ and it becomes optimistic if it was not. Otherwise, the third step of agent $X$ corresponds to the execution of $\mathcal{A}(1)$ and the agent stays pessimistic.

By doing so, we made a significant move forward. To understand why, we want to invite the reader to reconsider the case when there is at least one good agent that is optimistic at the beginning of a second step. As we have seen earlier, at the end of this second step, all good agents are necessarily gathered and every optimistic agent knows it. In view of the last changes made to our solution, when starting the third step, every good agent is then optimistic. As explained above the absence of pessimistic good agent is very helpful, and using here the same arguments, we are sure that when finishing the third step, all good agents are gathered and every good agent knows it because all of them are optimistic. Actually, it is even a little more subtle: the optimistic agents of the first generation (i.e., those that {were} already optimistic when starting the second step) know that the gathering is done and know that every good agent knows it. Concerning the optimistic agents of the second generation (i.e., those that became optimistic only when starting the third step), they just know that the gathering is done, but do not know whether the other agents know it or not. Recall that to get all good agents declare simultaneously the gathering achieved, we want to reach a round in which every good agent knows that every good agent knows that gathering is done. We are very close to such a consensus. To reach it, at the end of a third step, the optimistic agents of the first generation make themselves known to all agents. Note that if there were at least $f+1$ agents declaring to be optimistic agents of the first generation and if $f$ was part of $\mathcal{GK}$, the consensus would be reached. Indeed, among the agents declaring to be optimistic of the first generation, at least one is {necessarily} good and every agent can notice it: at this point we can show that every good agent knows that every good agent knows that gathering is done. 

However, the agents do not know $f$. That being said, at the end of a third step, note that an optimistic agent knowing that the gathering is done can compute an approximation $\tilde{f}$ of the number of Byzantine agents. More precisely, if the number of agents gathered in its node is $\sxo$, the optimistic agent knows than the number of Byzantine agents cannot exceed $\tilde{f}=\svS{\sxo}{\sxq}$ according to the definition of a strong team. Based on this fact, we are saved. Indeed, our algorithm is designed in such a way that all good agents correctly declare the gathering is achieved in the same round after having computed the same approximation $\tilde{f}$ and noticed at least $\tilde{f}+1$ agents that {claim} being optimistic of the first generation during a third step. We show that such an event necessarily occurs before any agent finishes the $(4|l_{min}|+8)$-th bit process of its transformed label, which permits to obtain the promised polynomial complexity. This
is where our feat of strength is: obtaining such a complexity with a small amount of global knowledge, while ensuring that the Byzantine agents cannot {confuse} the good agents in any way. Actually, our algorithm is judiciously orchestrated so that the only thing Byzantine agents can really do is just to accelerate the resolution of the problem.

		\subsection{Formal description}

Algorithm~\ref{onlyn:algo:main} gives the formal description of \sth{\swF}. As mentioned at the beginning of this section, we assume that $\svH=\svN$. \stH{\swF} uses the two building blocks \swA{} and \swD{} described in the previous section. It also uses two small subroutines, \swC{} and \swE{}, which are described after Algorithm~\ref{onlyn:algo:main}. Both these subroutines do not have any input parameters, but when executing them, the agent can access to the current value of every variable defined in Algorithm~\ref{onlyn:algo:main}. Hence the variables defined in Algorithm~\ref{onlyn:algo:main} can be viewed as variables of global scope.
		
			


			\begin{algorithm}
				\caption{\stH{\swF} executed by an agent $\sxa$ with label $\sxb$.}
				\label{onlyn:algo:main}
				\begin{footnotesize}
				\begin{algorithmic}[1]
					\State $N \gets 2^{(2^{\svH})}$\label{onlyn:line:1}
					\State Let $\svO{\sxb} = \sva_1 \ldots \sva_\svb$
					\State // Recall that $\svO{\sxb}$ is the transformed label of agent $\sxa$ (refer to Section~\ref{sec:pre})
					\State $\svd \gets 1$
					\State $\svc \gets 1$
					\State Execute \svF{\svD} \label{alg:explo}
					\While {$\svc \leq 3\svb$}\label{alg:deb}
						\If {$\svc \bmod 3 = 1$} \label{onlyn:line:test}
							\State $\sve \gets 0$ \label{onlyn:line:init}
							\State $\svf \gets \sva_{(\svc \bdiv 3) + 1}$
						\EndIf
						\State Execute \swAA{\svG{\svD}}{\svD}{\svf}\label{alg:group}
						\State Execute \swDD{\svG{\svD} + \swAa{\svD}}{\svD}\label{alg:merge}
						\State Execute \swCC{}
						\State Let $(\svf, \svd)$ be the value returned by \swCC{}
						\If {$\svf = \svK$}
							\State $\sve = \sve + 1$ \label{onlyn:line:incr}
						\EndIf
						\If {$\svc \bmod 3 = 0$}
							\State Execute \swEE{}
							\State Let flag be the boolean value returned by \swEE{}
							\If {flag=\svL{}}
								\State Declare that gathering is achieved \label{onlyn:line:exit}\label{alg:declare}
							\EndIf
						\EndIf
						\State Let $r$ be the time elapsed since the beginning of the execution of this procedure
						\State Wait $\svG{\svD} + \svc(3\svG{\svD} + 4(\swAa{\svD} + \swDd{\svD})+ 2) - r$ rounds \label{onlyn:algo:main:wait}
						\State $\svc \gets \svc + 1$
					\EndWhile \label{alg:fin}
				\end{algorithmic}
				\end{footnotesize}
			\end{algorithm}
		
When presenting the high level idea of our solution in the previous subsection, we used some qualifiers like ``experienced and in doubt'', ``optimistic of the second generation'' or ``optimistic of the first generation''. These qualifiers were only used to help the reader understand the essence of our solution and they do not appear explicitely in the formal description. To ease the transition from the high level idea, just note that these qualifiers are reflected in the values $1,2$ or $3$ of variable $\sve$. For example, an optimistic agent of the first generation corresponds to an agent for which $\sve=3$.		

			Now we give the formal descriptions of the subroutines \swC{} and \swE{}. Let us begin by subroutine \swC{}.

{\bf Subroutine \swCC{}}\\When executing this subroutine, an agent $\sxa$ can transit to different states that are \swT{}, \swU{} and \swV{}. The initial state is \swT{}. During an execution of this procedure, $\sxa$ never moves. Let us denote by $\sxl$ the node occupied by the agent while executing this subroutine and by $\sxf$ the value $\svG{\svD} + \swAa{\svD} + \swDd{\svD}$.
				
				\stb{\swT} Agent $\sxa$ spends one round in this state. Let $\sxm$ be the maximum number of agents {in \sta{\swT}} (including itself) that $\sxa$ notices at this round in node $\sxl$. Let $\svm$ be $max(\svd,\sxm)$. The agent $\sxa$ transits either to \sta{\swU} or to \sta{\swV}. It transits to \sta{\swU} if $\sve \ne 0$, or $2\svc > 3\svb$ and $\sxm \geq \svT{\svm}{\sxn}$. Otherwise, it transits to \sta{\swV}.
				
				\stb{\swU} Agent $\sxa$ waits $3\sxf$ rounds in this state. At the end of this waiting period, the agent exits the execution of \swCC{}: the returned value of the subroutine is then the couple $(\svK, \svm)$.
				
				\stb{\swV} Agent $\sxa$ waits $3\sxf$ rounds in this state. At the end of the waiting period, the agent exits the execution of \swCC{} and returns a couple, the value of which is as follows. If during the waiting period, agent $\sxa$ notices $2\sxf$ consecutive rounds such that in each of them there is at least one agent in \sta{\swU} in node $\sxl$, then the returned value is $(\svK, \svm)$. Otherwise the returned value is $(\svJ, \svm)$.

Now, let us describe the second subroutine \swEE{}.

{\bf Subroutine \swEE{}}~\\
Agent $\sxa$ waits a single round and then exists the execution of \swEE{}. During this round, agent $\sxa$ transmits the value of its variable $\sve$, and the word \swW{} in order to indicate that it is executing the same named subroutine. Let us denote by $\sxo$ the number of agents in its current node during the single round of the execution. If the value of the variable $\sve$ of $\sxa$ belongs to $\{2; 3\}$, and there are more than $\svS{\sxo}{\sxq}$ distinct agents transmitting $3$ and \swW{}, the subroutine returns \svL{}. Otherwise, the subroutine returns \svM{}.
				
		\subsection{Proof and analysis}
		
			In this subsection, we prove the correctness and the polynomiality of Algorithm~\swF{} to solve $\svA$-Byzantine gathering with strong teams in all graphs of size at most $\svC$, assuming that $\svH=\svN$. We start with the following proposition.
			
			\begin{proposition} \label{onlyn:proof:figures}
				Let $\szw$ be a positive integer. If within a team of $\szw$ agents, there are $\szx \geq \svR{\svA}$ good agents and at most $\svA$ Byzantine agents, then we have $\svA \leq \svS{\szw}{\szy} < \szx$.
			\end{proposition}
			
			\begin{proof}
				First of all, $\svA \leq \svS{\szw}{\szy}$ follows from the fact that $\szw \geq \svR{\svA}$. Then, let us assume by contradiction that $\svS{\szw}{\szy} \geq \szx$. This implies that $\svR{\szx} \leq \szw$. However, $2\szx < \svR{\szx}$ and $\szw \leq \szx + \svA < 2\szx$. By transitivity, we get $2\szx < 2\szx$. This is a contradiction, which completes the proof.
			\end{proof}

			The executions of all the subroutines and building blocks that are mentioned in the following statements and their proofs always occur during an execution of \sth{\swFF{\svN}} by an agent in a graph of size at most $\svC$. Hence, for ease of reading, we omit to mention it. Before going further, we give some notations that are used in the statement of the next lemma and its proof. For any good agent $\svk$, we denote by $\sxx{\svk}{\sxr}$ the round (if any) at which $\svk$ starts its $\sxr$-th execution of \sth{\swA}. We also denote by $\sua{\sxr}$ the first round (if any) at which there is at least one good agent that starts its $\sxr$-th execution of \sth{\swA}. Finally, according to line~\ref{onlyn:line:1} of Algorithm~\ref{onlyn:algo:main}, $N$ is the value $2^{(2^{\svH})}$.
			
			
			\begin{lemma} \label{onlyn:proof:delay}
				 Let $A$ be a good agent. For any positive integer $\sxr$, $\sxx{\svk}{\sxr + 1} = \sxx{\svk}{\sxr} + 3\svG{\svD} + 4\swAa{\svD} + 4\swDd{\svD} + 2$, and every good agent that starts its $\sxr$-th execution of \swA{} does it at round $\sua{\sxr} + \svG{\svD} - 1$ at the latest.
			\end{lemma}
			
			\begin{proof}
				 Since \svF{\svD} allows to visit every node of the graph, once the first awoken good agent has completed its first execution of \svF{\svD} at the beginning of \stg{\swF}, each good agent is awoken and has at least started its first execution of \svF{\svD}. Every good agent spends exactly $\svG{\svD}$ rounds executing it, and then starts its first execution of \sth{\swA}. Hence, every good agent starts its first execution of \sth{\swA} in some interval of $\svG{\svD}$ rounds, between rounds $\sua{1}$ and $\sua{1} + \svG{\svD} - 1$.
				
				We now consider the routines a good agent executes between the beginnings of any two consecutive executions of \swA{}. Let us show that their execution lasts at most $3\svG{\svD} + 4\swAa{\svD} + 4\swDd{\svD} + 2$ rounds. Any good agent spends at most $\swAa{\svD}$ rounds executing \swAA{\svG{\svD}}{\svD}{\svX} for any $\svX \in \{0; 1\}$, at most $\swDd{\svD}$ rounds executing \swDD{\svG{\svD} + \swAa{\svD}}{\svD}, exactly $3\sxf + 1$ rounds executing \swCC{}, and exactly 1 round executing \swEE{}. The sum of these amounts of rounds is $2 + 3\sxf + \swDd{\svD} + \swAa{\svD}$.
				
				In view of line~\ref{onlyn:algo:main:wait} of Algorithm~\ref{onlyn:algo:main} and since each agent spends exactly $\svG{\svD}$ rounds executing the initial \svF{\svD}, for any positive integer $\sub$ and any good agent $\szg$, the above sum is exactly the amount of rounds between $\sxx{\szg}{\sub}$ and $\sxx{\szg}{\sub + 1}$. Hence, any good agent spends exactly $4\swDd{\svD} + 4\swAa{\svD} + 3\svG{\svD} + 2$ rounds between the beginnings of any two consecutive executions of \sth{\swA}, which completes the proof.
			\end{proof}

			\begin{proposition} \label{onlyn:proof:prop}
				Consider a round $r$ at which two good agents $\svk$ and $\svl$ execute the same routine $\sxv$ from the set $\{$\swA; \swC; \swD; \swE$\}$. If $\svk$ is executing its $\sxr$-th execution of $\sxv$ at round $r$, then $\svl$ is also executing its $\sxr$-th execution of $\sxv$ at round $r$.
			\end{proposition}

			\begin{proof} 
				Let us assume by contradiction that there exists some round $\sxw$ at which two good agents $\sxy$ and $\szf$ are respectively executing their $\svj$-th and $\suf$-th execution of a same routine $\sxv$ from the set $\{$\swA; \swC; \swD; \swE$\}$ with $\svj < \suf$. In view of Lemma~\ref{onlyn:proof:delay}, we know that $(\sxx{\szf}{\suf}-\sxx{\szf}{\svj})\geq 3\svG{\svD} + 4\swAa{\svD} + 4\swDd{\svD} + 2$. Thus $(\sxw - \sxx{\szf}{\svj})\geq 3\svG{\svD} + 4\swAa{\svD} + 4\swDd{\svD} + 2$. However, if $\sxv$ is \swA{} or \swD{}, we know that $(\sxw - \sxx{\sxy}{\svj})\leq\swAa{\svD} + \swDd{\svD}$. Hence, $(\sxx{\sxy}{\svj}-\sxx{\szf}{\svj})>\svG{\svD}$, which contradicts Lemma~\ref{onlyn:proof:delay}.
				
				Hence, $\sxv$ must be \swC{} or \swE{}. In either case, in view of Algorithm~\ref{onlyn:algo:main}, between $\sxx{\szf}{\suf}$ and $\sxw$, $\szf$ must have executed \swA{} and \swD{}. Executing these routines requires a minimal number a rounds which is strictly larger than $\svG{\svD}$. This means that, $(\sxw - \sxx{\szf}{\svj})\geq4\svG{\svD} + 4\swAa{\svD} + 4\swDd{\svD} + 3$. On the other hand, $(\sxw - \sxx{\sxy}{\svj})\leq 3\svG{\svD} + 4\swAa{\svD} + 4\swDd{\svD} + 2$ rounds. Hence, $(\sxx{\sxy}{\svj}-\sxx{\szf}{\svj})>\svG{\svD}$, which contradicts again Lemma~\ref{onlyn:proof:delay}.
			\end{proof}

			By Lemma~\ref{onlyn:proof:delay} and lines~\ref{alg:group}-\ref{alg:merge} of Algorithm~\ref{onlyn:algo:main}, we know that all good agents that start their $\sxr$-th execution of \swDD{\svG{\svD} + \swAa{\svD}}{\svD} for a given positive integer $i$, do it in an interval lasting at most the number of rounds given as first parameter of \swD.  Hence, in view of Theorem~\ref{Block:theo:Principal-2}, Theorem~\ref{block2:proof:theo}, and Lemma~\ref{onlyn:proof:delay}, we have the following corollary.

			\begin{corollary} \label{onlyn:proof:cor}
				Assume that for a given positive integer $i$, there is a group of at least $\svR{\svA}$ good agents that start (at possibly different nodes or rounds) their $\sxr$-th executions of \swAA{\svG{\svD}}{\svD}{\svf} with $\svf=0$ for at least one good agent and $\svf=1$ for at least one other good agent. There exist a node $\sxt$ and a round $\sxu$ such that each good agent in the graph that completes its $\sxr$-th execution of \swD{} does it at round $\sxu$ in node $\sxt$.
			\end{corollary}
			

Before proving Theorem~\ref{onlyn:proof:live} that is the main result of this section, we still need to prove the following series of four lemmas.
			\begin{lemma} \label{onlyn:proof:lemma2}
				Assume that for a given positive integer $\sxr$, there are at least $5\svA + 2$ good agents that start (at possibly different rounds) their $\sxr$-th executions of \swAA{\svG{\svD}}{\svD}{\svf} in the same node $\sxt$ with $\svf=0$. There exist a round $\sae$ and a node $\suc$ such that each good agent in the graph that completes its $\sxr$-th execution of \sth{\swD} does it at round $\sae$ in node $\suc$.
			\end{lemma}
			
			\begin{proof}
				Assume that there exist an integer $\sym$ and a node $\syn$ such that a group $\syl$ of at least $5\svA + 2$ good agents all start their $\sym$-th executions of \swAA{\svG{\svD}}{\svD}{\svf} in the same node $\sxt$ with $\svf=0$. In view of Lemma~\ref{onlyn:proof:delay}, every good agent that starts its $\sym$-th execution of \swA{}, does it between rounds $\sua{\sym}$ and $\sua{\sym} + \svG{\svD} - 1$. In view of the description of \swA{}, at the beginning of its execution of \swAA{\svG{\svD}}{\svD}{\svK}, every good agent, which is called a \swN{}, first enters \sta{\swO} and spends strictly more than $\svG{\svD}$ rounds waiting in this state. Hence, there is at least one round at which each good agent of $\syl$ is waiting in \sta{\swO} in node $\syn$ during the first phase of its $\sym$-th execution of \swA{}. Thus, by Lemma~\ref{Block:lem:groupx-f}, at least $4\svA + 2$ good agents exit their $\sym$-th execution of \swA{} at the same round and in the same node. This means that at least $4\svA + 2$ good agents start their $\sym$-th execution of \swDD{\svG{\svD} + \swAa{\svD}}{\svD} at the same round and in the same node. Besides, each good agent spends at most $\swAa{\svD}$ rounds in any execution of \swA{} which means that each good agent that completes its $\sym$-th execution of \swA{} does it between round $\sua{\sym}$ and round $\sua{\sym} + \swAa{\svD} + \svG{\svD} - 1$. Hence, in view of Theorem~\ref{block2:proof:theo}, there exist a round $\sad$ and a node $\sud$ such that each good agent that completes its $\sym$-th execution of \swDD{\svG{\svD} + \swAa{\svD}}{\svD} does it at $\sad$ in $\sud$.
			\end{proof}
			
			\begin{lemma} \label{onlyn:proof:1node}
				Assume that $g\geq\svR{\svA}$ good agents start (at possibly different nodes or rounds) their $(3\syo + 1)$-th execution of \stj{\swA}, for a given integer $\syo$. Let $\syp \leq 3$ be the smallest positive integer, if any, such that the $(3\syo + \syp)$-th execution of \sth{\swC} by at least one good agent returns a couple whose the first element is $\svK$. There exists a node $\sxt$ such that each good agent that enters \sta{\swU} during its $(3\syo + \syp)$-th execution of \swC{}, does it in $\sxt$.
			\end{lemma}
			
			\begin{proof}
				Let us assume by contradiction that two good agents $\syq$ and $\syr$ both enter \sta{\swU} during their $(3\syo + \syp)$-th executions of \swC{} but from different nodes, respectively $\sys$ and $\syt$. To conduct this proof, we need to explain what the entrance of these good agents in \sta{\swU} implies. Note that during their $(3\syo + \syp)$-th executions of \swC{}, both agents have the same value for variable $\svc$ \ie $(3\syo + \syp)$. In view of lines~\ref{onlyn:line:test}-\ref{onlyn:line:init} of Algorithm~\ref{onlyn:algo:main}, when a good agent starts its $(3\syo + 1)$-th execution of \swA{}, the value of its variable $\sve$ is 0. Since $\sve$  is only incremented on line~\ref{onlyn:line:incr}, when the first element of the pair returned by \swC{} is $\svK$, we know by definition of $\syp$ that at the beginning of their $(3\syo + \syp)$-th executions of \swC{}, the value of $\sve$ for both $\syq$ and $\syr$ is still 0. Thus, in view of the description of \sth{\swC}, $2(3\syo + \syp) > 3\sza$ (resp. $2(3\syo + \syp) > 3\szb$) and while in \sta{\swT}, $\syq$ (resp. $\syr$) notices at least $\svT{\syw}{\syy}$ (resp. $\svT{\syx}{\syz}$) agents {in \sta{\swT}} in its node, where $\sza$ (resp. $\szb$) denotes the length of the transformed label of $\syq$ (resp. $\syr$) as defined in Algorithm~\ref{onlyn:algo:main} and $\syw$ (resp. $\syx$) denotes the value of $\svm$ that is used in \sta{\swT} by agent $\syq$ (resp. $\syr$).
				
				We now explain what $2(3\syo + \syp) > 3\sza$ and $2(3\syo + \syp) > 3\szb$ imply. By Proposition~\ref{prelim:label}, it means that there exists a positive integer $\szc \leq \syo$ such that the $\szc$-th bits in the transformed labels of $\syq$ and $\syr$ are different. In view of Algorithm~\ref{onlyn:algo:main}, this means that for their $(3\szc - 2)$-th executions of \swA{}, one of them executes \swAA{\svG{\svD}}{\svD}{0} while the other one executes \swAA{\svG{\svD}}{\svD}{1}. Hence, in view of Corollary~\ref{onlyn:proof:cor}, there exist a node $\sue$ and a round $\sxi$ such that each good agent that completes its $(3\szc - 2)$-th execution of \swD{}, does it in $\sue$ at round $\sxi$. This means that each good agent that starts its $(3\szc - 2)$-th execution of \swC{}, and thus enters \sta{\swT}, does it in $\sue$ at round $\sxi + 1$. By assumption, these good agents are at least $\szx \geq \svR{\svA}$ and among them, there are $\syq$ and $\syr$. This means that the number of agents {in \sta{\swT}} that $\syq$ (resp. $\syr$) notices during its $(3\szc - 2)$-th execution of \swC{}, and thus $\syw$ (resp. $\syx$) is at least $\szx$.
				
				Let us now give the consequences of the fact that $\syq$ (resp. $\syr$) notices at least $\svT{\syw}{\syy}$ (resp. $\svT{\syx}{\syz}$) agents {in \sta{\swT} while in the same state} during its $(3\syo + \syp)$-th execution of \swC{}. In view of Proposition~\ref{onlyn:proof:figures}, both $\svS{\syw}{\syy}$ and $\svS{\syx}{\syz}$ are at least $\svA$. Let us assume without loss of generality that $\svT{\syx}{\syz}$ is at least $\svT{\syw}{\syy}$. Hence, the sum of the numbers of agents {in \sta{\swT}} noticed by $\syq$ or $\syr$ {while in the same state} during their $(3\syo + \syp)$-th execution of \swC{} is at least $2(\svT{\syw}{\syy})$ \ie at least $\szx + \syw - 2\svS{\syw}{\syy}$. Besides, $\syw - 2\svS{\syw}{\syy}$ is greater than $2\svS{\syw}{\syy} \geq 2\svA$. This means that the total number $q$ of agents {in \sta{\swT}} noticed by $\syq$ or $\syr$ {while in the same state} during their $(3\syo + \syp)$-th executions of \swC{} is greater than $\szx + 2\svA$. However, this is impossible as explained in the next paragraph. 

In view of Proposition~\ref{onlyn:proof:prop}, when $\syq$ or $\syr$ starts its $(3\syo + \syp)$-th execution of \swC{}, each good agent in \sta{\swT} is also starting its $(3\syo + \syp)$-th execution of \swC{}. No good agent can be in \sta{\swT} during its $(3\syo + \syp)$-th execution of \swC{} both in $\sys$ and in $\syt$, which means that among the $q>\szx + 2\svA$ agents noticed by $\syq$ or $\syr$, at most $\szx$ are good. However, in every round, there {are} at most $\svA$ Byzantine agents in $\sys$ or $\syt$. Hence, $q$ cannot be greater than $\szx + 2\svA$: this leads to a contradiction that proves the theorem.
			\end{proof}
			
			\begin{lemma} \label{onlyn:proof:lemma3}
				Let $\sxr$ and $\sxs$ two integers such that $\sxs\in\{1;2\}$. Assume that at least $\svR{\svA}$ good agents start (at possibly different nodes or rounds) their $(3\sxr + 1)$-th executions of \stj{\swA}. If the $(3\sxr + \sxs)$-th execution of \stj{\swC} by at least one good agent returns a couple whose the first element is $\svK$, then for every integer $\sxs < \sug \leq 3$, there exist a round $\sui$ and a node $\suh$ such that every $(3\sxr + \sug)$-th execution of \swC{} by any good agent finishes at round $\sui$ in node $\suh$ and returns a pair whose first element is $\svK$.
			\end{lemma}
			
			\begin{proof}
				To prove this lemma, it is enough to show that for all integers $\szd$ and $\szz$ such that $\szz\in \{1; 2\}$, if at least $\svR{\svA}$ good agents start (at possibly different nodes or rounds)  their $(3\szd + 1)$-th executions of {\swA{}}, and the $(3\sxr + \szz)$-th execution of \stj{\swC} by at least one good agent $\sze$ returns a couple whose the first element is $\svK$, then we have the following property: there exist a round $\sui$ and a node $\suh$ such every $(3\szd + \szz + 1)$-th execution of \swC{} by any good agent finishes in node $\suh$ at round $\sui$, and returns a couple whose the first element is $\svK$.

				We consider three cases. In the first case, $\szz = 1$. In the second case, $\szz = 2$ and there is no good agent whose $(3\szd + 1)$-th execution of \swC{} returns a couple in which the first element is $\svK$. In the third case, $\szz = 2$ and there is at least one good agent $\sra$ (not necessarily different from $\sze$) whose $(3\szd + 1)$-th execution of \swC{} returns a couple in which the first element is $\svK$.


				Let us consider the first case. We first prove that all the good agents which complete their {$(3\szd + 2)$-th} execution of \swD{} do it at the same round, and in the same node. In view Algorithm~\ref{onlyn:algo:main}, {during the $(3\szd + 2)$-th execution of \swAA{\svG{\svD}}{\svD}{\svf} by agent $\sze$} we have $\svf=\svK$. So, if there exists a good agent that uses $\svJ$ for parameter $\svf$ during its {$(3\szd + 2)$-th} execution of \swAA{\svG{\svD}}{\svD}{\svf}, then in view of Corollary~\ref{onlyn:proof:cor}, all the good agents that complete their {$(3\szd + 2)$-th} executions of \swD{} do it at the same round and in the same node. The situation, in which each good agent that starts its {$(3\szd + 2)$-th} execution of \swAA{\svG{\svD}}{\svD}{\svf} uses $\svf=\svK$, is a little trikier to analyse. In this situation, in view of Algorithm~\ref{onlyn:algo:main}, this means that every {$(3\szd + 1)$-th execution of \swC{}} by any good agent $\srb$ returns a couple whose first element is $\svK$. Thus, during this execution agent $\srb$ either enters \sta{\swU}, or while in \sta{\swV} it notices at least one agent in \sta{\swU} for at least $2\sxf$ consecutive rounds. Moreover, in view of Lemma~\ref{onlyn:proof:delay} and the definitions of values $\swAa{\svD}$ and $\swDd{\svD}$, every good agent that starts its {$(3\szd + 1)$-th} execution of \swC{}, does it between rounds $\sua{3\szd + \szz}$ and $\sua{3\szd + \szz} + \swAa{\svD} + \swDd{\svD} + \svG{\svD} - 1 = \sua{3\szd + \szz} + \sxf - 1$ ($\sxf$ is defined in the description of \swC{}). According to the description of \sta{\swT} (resp. \sta{\swV}), every good agent that enters (resp. exits) \sta{\swV}, does it between rounds $\sua{3\szd  + \szz} + 1$ and $\sua{3\szd + \szz} + \sxf$ (resp. $\sua{3\szd + \szz} + 3\sxf$ and $\sua{3\szd + \szz} + 4\sxf - 1$). Hence, there are at most $4\sxf - 1$ rounds at which at least one good agent is in \sta{\swV} during its {$(3\szd + 1)$-th} execution of \swC{}. This implies there is at least one round $\src$ that overlaps all the intervals of $2\sxf$ rounds noticed by the good agents in \sta{\swV}, and such that in round $\src$ each good agent in \sta{\swV} is in the same node as at least one agent in \sta{\swU}. By Lemma~\ref{onlyn:proof:1node}, there is at most one node where good agents can enter \sta{\swU} during their {$(3\szd + 1)$-th} executions of \swC{}. This implies that there are at most $\svA + 1$ nodes in the graph from which any good agent can exit \sta{\swT} during its {$(3\szd + 1)$-th} execution of \swC{}. Since by assumption at least $\svR{\svA}$ good agents execute their {$(3\szd + 1)$-th} executions of \swC{}, we have the following: there exists at least one node $\srd$ such that at least $5\svA + 2$ good agents are in $\srd$ during their {$(3\szd + 1)$-th} executions of \swC{}  {as well as at the beginning of their $(3\szd + 2)$-th} executions of \swAA{\svG{\svD}}{\svD}{\svf}. Recall that these good agents all use $\svK$ as value for parameter $\svf$ during their {$(3\szd + 2)$-th} executions of \swAA{\svG{\svD}}{\svD}{\svf}. Hence, by Lemma~\ref{onlyn:proof:lemma2} it follows that all the good agents that complete their {$(3\szd + 2)$-th} executions of \swD{} do it at the same round and in the same node.

				So, in the first case, we know that all the good agents that complete their {$(3\szd + 2)$-th} executions of \swD{} do it at the same round and in the same node. Now, we show that these agents all complete at the same round and in the same node their {$(3\szd + 2)$-th} executions of \swC{} that returns a couple whose the first element is $\svK$. In view of the description of \swC, these good agents do not move and spend exactly $3\sxf + 1$ rounds during their {$(3\szd + 2)$-th} executions of \swC. Hence, they enter and exit \sta{\swT} at the same round, and complete their {$(3\szd + 2)$-th} executions of \swC{} at the same round and in the same node. Moreover, since the {$(3\sxr + 1)$-th} execution of \stj{\swC} by agent $\sze$ returns a couple whose the first element is $\svK$, we know that during its {$(3\szd + 2)$-th} of \stj{\swC}, the value of its variable $\sve$ is different from $0$ and it enters \sta{\swU}. Hence, each good agent that enters \sta{\swV} during its {$(3\szd + 2)$-th} execution of \swC{} notices agent $\sze$ in \sta{\swU} during $3\sxf \geq 2\sxf$ consecutive rounds. As a result, there exist a round $\sui$ and a node $\suh$ such that every {$(3\szd + 2)$-th} execution of \swC{} by any good agent finishes in node $\suh$ at round $\sui$, and returns a couple whose the first element is $\svK$.
				
				Using similar arguments to those used in the first case, we can show in the second case that there exist a round and a node in which every $(3\szd + 3)$-th execution of \swC{} by any good agent finishes and returns a couple whose the first element is $\svK$.
				
				Let us now consider the {third case \ie $\szz = 2$ and there is at least one good agent $\sra$ whose $(3\szd + 1)$-th execution of \swC{} returns a couple in which the first element is $\svK$.} Using similar arguments to those used in the first case, we can show that, there exist a round and a node in which every {$(3\szd + 2)$-th} execution of \swC{} by any good agent finishes and returns a couple whose first element is $\svK$. All these good agents start their {$(3\szd + 3)$-th} executions of \swAA{\svG{\svD}}{\svD}{\svf} from the same node with $\svf=\svK$. In view of Lemma~\ref{onlyn:proof:lemma2}, this implies that there exist a round $\sre$ and a node $\srf$ such that every {$(3\szd + 3)$-th} executions of \swD{} by any good agent finishes in node $\srf$ at round $\sre$. Moreover, since every {$(3\szd + 2)$-th} execution of \swC{} of any good agent $\srb$ returns a couple whose the first element is $\svK$, during its {$(3\szd + 3)$-th} execution of \swC{}, agent $\srb$ enters \sta{\swU}: this execution of \swC{} by agent $\srb$ lasts exactly $3\sxf + 1$ rounds during which it does not move from $\srf$. Hence, the {$(3\szd + 3)$-th} execution of \swC{} of agent $\srb$ returns a couple whose the first element is $\svK$ at round $\sre+3\sxf + 1$, which completes the proof.
			\end{proof}

			\begin{lemma} \label{onlyn:proof:safe}
				Assume there is a group $G$ of at least $\svR{\svA}$ good agents executing Algorithm~\swF{} at a round $\sxu$. If at least one agent of $G$ declares that gathering is achieved at round $\sxu$ in a node $\sxt$, then all agents of $G$ declare that gathering is achieved at $\sxu$ in $\sxt$.
			\end{lemma}
			
			\begin{proof}
				 By assumption, there is at least one good agent $\szo$ that declares that gathering is achieved at $\szn$. Let $\szp$ be the value of the variable $\svc$ of agent $\szo$ at round $\szn$. In view of Algorithm~\ref{onlyn:algo:main}, it declares that gathering is achieved after executing \stj{\swE}, and there exists an integer $\suj$ such that $\szp = (3\suj + 3)$.
				
				In view of subroutine \swE{}, since $\szo$ declares gathering achieved at round $\szn$, the value of its variable $\sve$ is either 2 or 3. In view of Algorithm~\ref{onlyn:algo:main}, this means that there are either two or three executions of \swC{}, out of the three since the beginning of the $(3\suj + 1)$-th execution of \swA{} by agent $\szo$, which have returned a couple whose the first element is $\svK$. In view of Lemma~\ref{onlyn:proof:lemma3}, this means that there exist a round $\sac$ and a node $\szq$ such that each agent of $G$ completes at $\sac$ in $\szq$ its $(3\suj + 3)$-th execution of \swC{}, the returned value of which is a couple whose the first element is $\svK$.

Consider the set of the values of variable $\sve$ of every good agent of $G$ at the end of its $(3\suj + 3)$-th execution of \swC{}, and denote by $\szr$ the maximum one. Since at the end of the $(3\suj + 3)$-th execution of \swC{} by agent $\szo$, the variable $\sve$ of $\szo$ is either $2$ or $3$, we know that $\szr\in\{2;3\}$. Lemma~\ref{onlyn:proof:lemma3} implies that at the end of the $(3\suj + 3)$-th execution of \swC{} by every good agent $B$ of $G$ (including $\szo$), the variable $\sve$ of $B$ is either $\szr$ or $\szr - 1$.
				
				Each agent of $G$ starts its {$(\suj + 1)$-th execution} of subroutine \swE{} at $\sac + 1$ in $\szq$. According to the description of this subroutine and Algorithm~\ref{onlyn:algo:main}, agent $\szo$ declares that gathering is achieved at round $\szn$ because at the previous round, while executing \swE{}, it notices strictly more than $\svS{\sxo}{\sxq}$ distinct agents executing the same procedure and transmitting $3$. Thus, the round at which all the agents of $G$ execute \swEE{} in $\szq$ is $\szn - 1$. Since at least $\svR{\svA}$ good agents are in the same node, by Proposition~\ref{onlyn:proof:figures}, we know that at least one good agent transmits $3$ at round $\szn - 1$. In view of the fact that the integer transmitted by any good agent executing \swEE{} is the value of its variable $\sve$, we know that $\szr$ is 3 and the value of variable $\sve$ of each agent of $G$ is either $2$ or $3$. This means the execution of \swEE{} of every good agent of $G$ returns \svL{} at round $\szn - 1$ in $\szq$, and every good agent of $G$ declares that gathering is achieved at round $\szn$ with agent $\szo$, which completes this proof.
			\end{proof}
			

We are now ready to prove the final result of this section. Recall that a strong team is a team in which the number of good agents is at least $5f^2+6f+2$. As the reader would have noticed, a good agent can execute several iterations of the while loop of Algorithm~\ref{onlyn:algo:main} (cf. lines~\ref{alg:deb} to~\ref{alg:fin}): given a good agent $A$,  we will say that the $i$-th iteration of this while loop by agent $A$ is of \emph{order $i$}.

			
			\begin{theorem} \label{onlyn:proof:live}
Assuming that $\svH=\svN$, Algorithm~\swF{} solves $\svA$-Byzantine gathering with every strong team in all graph of size at most $\svC$, and has a time complexity that is polynomial in $\svC$ and $|\svB|$.
			\end{theorem}
	
			\begin{proof}
				Let $\srh$ be the first round in which a good agent finishes the execution of Algorithm~\swF{}. Since, the adversary wakes up at least one good agent, we know that round $r$ exists. Since $\svH=\svN$, we know that $\svD=2^{(2^{\svH})}$ is at least $\svC$, and thus according to line~\ref{alg:explo} of Algorithm~\ref{onlyn:algo:main}, all the good agents are executing Algorithm~\swF{} at round $\srh$. As a result, in view of Lemma~\ref{onlyn:proof:safe}, we just have to prove the following two properties to state that the theorem holds. The first property is that there exists at least one good agent that declares gathering is achieved at round $\srh$ (note that although we will show in the sequel that it is impossible, we cannot rule out for now the possibility that an agent might finish the execution of Algorithm~\swF{} without declaring gathering is achieved). The second property is that at round $\srh$, the first woken-up agent (or one of the first, if there are several such agents) has spent a time that is at most polynomial in $\svC$ and $|\svB|$ to execute Algorithm~\swF{}.

				Let us first focus on the first property and consider the good agent $\szs$ with the smallest label $\svB$. Let $\alpha=3\svP{\svO{\svB}}$. In view of Algorithm~\ref{onlyn:algo:main}, each good agent executes at least $\alpha$ iterations of the while loop of Algorithm~\ref{onlyn:algo:main}, unless it declares that gathering is achieved before. We consider two cases: either there is at least one good agent $\srg$ that never starts executing its $\alpha$-th iteration of the while loop, or every good agent start executing at some point its $\alpha$-th iteration of the while loop.

Concerning the first case, assume without loss of generality that $\srg$ is the first agent that stops executing Algorithm $\texttt{GATHER}$ before starting its $\alpha$-th iteration of the while loop. According to Lemma~\ref{onlyn:proof:delay}, the time spent executing an iteration is the same regardless of the executing good agent and the order of the iteration, and this time is greater than the difference between the rounds at which any two good agents start iterations of the same order. Hence, when agent $\srg$ stops executing Algorithm $\texttt{GATHER}$, no good agent has completed its $\alpha$-th iteration of the while loop. This implies that agent $B$ finishes its execution of Algorithm $\texttt{GATHER}$ at round $r$. Moreover, the fact that $B$ stops executing Algorithm $\texttt{GATHER}$ before starting its $\alpha$-th iteration of the while loop, implies that $B$ declares the gathering is achieved at round $r$: this proves that the first property holds in the first case.

Let us move on to the second case. In view of Proposition~\ref{prelim:label}, for any given good agent $\szt$ different from $A$, there exist two positive integers $\szu$ and $\szv$ such that $2\szu \leq \svP{\svO{\svB}}$, $\svP{\svO{\svB}} < 2\szv \leq 2\svP{\svO{\svB}}$ and the $\szu$-th (resp. $\szv$-th) bits in the transformed labels of $\szs$ and $\szt$ are different. Hence, at round $\srh$, each good agent has at least started executing its $\alpha$-th iteration of the while loop, and thus has completed its $(3\szu - 2)$-th iteration and at least started its $(3\szv - 2)$-th iteration. 

Moreover, in view of Algorithm~\ref{onlyn:algo:main}, for its $(3\szu - 2)$-th (resp. $(3\szv - 2)$-th) execution of \swAA{\svG{\svD}}{\svD}{\svf}, agent $\szs$ uses for parameter $\svf$ a value belonging to $\{0;1\}$ that is different of that used by agent $\szt$ (which also belongs to $\{0;1\}$) during its $(3\szu - 2)$-th (resp. $(3\szv - 2)$-th) execution of \swAA{\svG{\svD}}{\svD}{\svf}. By Corollary~\ref{onlyn:proof:cor}, there exist a round $\szi$ and a node $\szh$ (resp. $\szj$ and $\szk$) such that each good agent completes its $(3\szu - 2)$-th (resp. $(3\szv - 2)$-th) execution of \swD{} at $\szi$ in $\szh$ (resp. at $\szj$ in $\szk$). At round $\szi + 1$, each good agent enters \sta{\swT} in node $\szh$. Thus, at this point the value of variable $\svd$ of each good agent is at least the number of good agents and at most the total number of agents. Since there are at least $\svR{\svA}$ good agents, in view of Proposition~\ref{onlyn:proof:figures}, $\svS{\svd}{\sxn}$ is at least $\svA$, and $\svT{\svd}{\sxn}$ is at most the number of good agents.
				
				Furthermore, the length of the transformed label of $\szs$ is $\svP{\svO{\svB}}$. This means that during its $(3\szv - 2)$-th execution of \swC{}, at round $\szj + 2$, while all good agents are in $\szk$, agent $\szs$ enters \sta{\swU}. At the same round, every other good agent is also in $\szk$ entering either \sta{\swU} or \sta{\swV}. Whichever the state, they spend $3\sxf$ rounds in it so that all the good agents in \sta{\swV} notice agent $\szs$ in \sta{\swU} during at least $2\sxf$ rounds. As a result, every good agent finishes its $(3\szv - 2)$-th execution of \swC{} that returns a pair whose first element is $\svK$. From Lemma~\ref{onlyn:proof:lemma3}, it follows that there exist a round $\szl$ and a node $\szm$ such that each good agent completes its $3\szv$-th execution of \stj{\swC} at $\szl$ in $\szm$, and the value of variable $\sve$ of each good agent at round $\szl$ is 3. From round $\szl+1$ on, each good agent starts its $\szv$-th execution of \swE{}. When executing this procedure, each of them transmits the word \swW{} and the value 3 of its variable $\sve$. In view of Proposition~\ref{onlyn:proof:figures}, there are strictly more than $\svS{\sxo}{\sxq}$ good agents. Hence, agent $\szs$ as well as all good agents return \svL{}, and thus declare that gathering is achieved at round $r=\szl + 2$ in node $\szm$ which proves that the first property holds in the second case as well.
				
				We now show the second property. According to the two cases analyzed above, the good agents declare that the gathering is achieved at round $r$ before any of them starts its iteration of the while loop of order $\alpha+1$: the value $\alpha$ is polynomial in $\svP{\svB}$ since $\alpha=3\svP{\svO{\svB}}$ and $\svP{\svO{\svB}} = 4\svP{\svB} + 8$. Besides, the number of rounds required to execute any iteration of the while loop is bounded by $4(\svG{\svD} + \swAa{\svD} + \swDd{\svD} + 1)$ in view of Lemma~\ref{onlyn:proof:delay}. Note that in view of the definitions of $\svG{\svD}$, $\swAa{\svD}$ and $\swDd{\svD}$, $4(\svG{\svD} + \swAa{\svD} + \swDd{\svD} + 1)$ is polynomial in $N$, and thus in $n$ as $N=2^{(2^{\lceil \log \log n \rceil})}$ (cf. line~\ref{onlyn:line:1} of Algorithm~\ref{onlyn:algo:main}). Hence, the total number of rounds spent by any good agent before round $\srh$ is bounded by $12(4\svP{\svB} + 8)(\svG{\svD} + \swAa{\svD} + \swDd{\svD} + 1)$, which is polynomial in $\svC$ and $\svP{\svB}$. This concludes the proof of the second property, and by extension, of the theorem. 

			\end{proof}

%% file: Sections/generalresults.tex
\section{The negative result}
\label{sec:gen}

Algorithm  {\tt GATHER} introduced in the previous section uses the value $\lceil \log \log n \rceil$ as global knowledge, which can be coded with a binary string of size $\mathcal{O}(\log \log \log n)$.
In this section, we show that, to solve Byzantine gathering with strong teams, in all graph of size at most $n$, in a time polynomial in $n$ and $|l_{min}|$, the order of magnitude of the size of knowledge used by our algorithm {\tt GATHER} is optimal. More precisely, we have the following theorem.






\begin{theorem}
There is no algorithm solving $f$-Byzantine gathering with strong teams for all $f$ and in all graphs of size at most $n$, which is polynomial in $n$ and $|l_{min}|$ and which uses a global knowledge of size $o(\log \log \log n)$.
\end{theorem}

\begin{proof}
Suppose by contradiction that the theorem is false. Hence, there exists an algorithm {\tt Alg} that solves $f$-Byzantine gathering with strong teams for all $f$ in all graphs of size at most $n$,
which is polynomial in $n$ and $|l_{min}|$ and which uses a global knowledge of size $o(\log \log \log n)$. The proof relies on the construction of a family $\mathcal{F}_{n}$ (for any $n\geq4$) of initial instances with strong teams such that for each of them the graph size is at most $n$. Our goal is to prove that there is an instance from $\mathcal{F}_{n}$ for which algorithm {\tt Alg} needs a global knowledge whose size does not belong to $o(\log \log \log n)$, which would be a contradiction with the definition of {\tt Alg}. Let us first present the construction of an infinite sequence of instances $\mathcal{I}=I_0,I_1,I_2,\ldots,I_i,\ldots$ by induction on $i$. Instance $I_0$ consists of an oriented ring of $4$ nodes (i.e., a ring in which at each node the edge going clockwise has port number $0$ and the edge going anti-clockwise has port $1$). In this ring, there is no Byzantine agent but there are two good agents labeled $0$ and $1$ that are placed in diametrically opposed nodes. All the agents in $I_0$ wake up at the same time.

Now let us describe the construction of instance $I_i$ with $i\geq1$ using some features of instance $I_{i-1}$. Let $c$ be the smallest constant integer such that the time complexity of algorithm {\tt Alg} is at most $n^c$ from every instance made of a graph of size at most $n$ with a strong team in which $|l_{min}|=1$. Let $\mu_{i-1}$ and $n_{i-1}$ be respectively the total number of agents in $I_{i-1}$ and the number of nodes in the graph of $I_{i-1}$. Instance $I_i$ consists of an oriented ring of $(n_{i-1})^{4c}$ nodes. In this ring an agent labeled $0$ is placed on a node denoted by $v_0$. In each of both nodes that are adjacent to $v_0$, $(n_{i-1})^c*\mu_{i-1}$ Byzantine agents are placed (which gives a total of $2*(n_{i-1})^c*\mu_{i-1}$ Byzantine agents). On the node that is diametrically opposed to $v_0$, enough good agents are placed in order to have a strong team. The way of assigning labels to all agents that are not at $v_0$ is arbitrary but respects the condition that initially no two agents share the same label. Finally, all the agents in $I_{i}$ wake up at the same time. This closes the description of the construction of $\mathcal{I}$, for which we have the following claim.

\noindent \textbf{Claim 1} For any two instances $I_j$ and $I_{j'}$ of $\mathcal{I}$, algorithm {\tt Alg} requires a distinct global knowledge.

\noindent \textbf{Proof of Claim 1}

Assume by contradiction that the claim does not hold for two instances $I_j$ and $I_{j'}$ such that $j<j'$. Consider any execution $EX_j$ of algorithm {\tt Alg} from $I_j$. According to the construction of $\mathcal{I}$, we know that every agent is woken up at the first round of $EX_j$. We denote by $r_1,r_2,\ldots,r_k$ the sequence of consecutive rounds from the first round of $EX_j$ to the round when all good agents declare that gathering is done. We also denote by $G_i$ the group of agents (possibly empty) that are with the good agent labeled $0$ at round $r_i$ of $EX_j$. Now, using execution $EX_j$, let us describe a possible execution $EX_{j'}$ of algorithm {\tt Alg} from $I_{j'}$: this execution is designed in such a way that it will fool the good agent labeled $0$ and will induce it into premature termination. According to the construction of $\mathcal{I}$, all the agents of $I_{j'}$ are woken up in the first round of $I_{j'}$ and all the good ones are executing algorithm {\tt Alg}. In the first round of $EX_{j'}$ the agent labeled $0$ is alone (as in the first round of $EX_j$). Then, for each $i\in{2,\ldots,k}$, the good agent labeled $0$ in $EX_{j'}$ meets a group of $|G_i|$ Byzantine agents whose the multiset of labels is exactly the same as the multiset of labels belonging to the agents of $G_i$ in the $i$th round of $EX_j$. This is always possible in view of the fact that for each $i\in{1,\ldots,k}$, $|G_i|\leq\mu_j$ and the Byzantine agents of $I_{j'}$ can choose to move by ensuring that in the $i$th round of $EX_{j'}$ it remains at least $(k-i)*\mu_j$ Byzantine agents in the node adjacent to the one occupied by the agent labeled $0$ in the clockwise direction (resp. anti-clockwise direction): indeed according to the construction of $I_{j'}$, in each of both nodes adjacent to the starting node of the good agent labeled $0$, there are initially  $(n_{j'-1})^c*\mu_{j'-1}\geq k*\mu_{j'-1}$ Byzantine agents, as $k\leq (n_{j})^c \leq (n_{j'-1})^c$. Finally, if algorithm {\tt Alg} prescribes some message exchange between agents during their meetings, then the Byzantine agents in execution $EX_{j'}$ give exactly the same information to $0$, as the agents with respective labels in execution $EX_j$. Hence, from the point of view of agent $0$, the first $k$ rounds of $EX_j$ look exactly identical to the first $k$ rounds of $EX_{j'}$. This is due to the actions of Byzantine agents, the fact that all nodes in $I_j$ and $I_{j'}$ look identical, and also because $k\leq (n_{j})^c$ which implies that, regardless of the algorithm {\tt Alg}, the agent labeled $0$ cannot meet any good agent in the first $k$ rounds of $EX_{j'}$ as the distance between agent $0$ and any other good agent is initially at least $\frac{(n_{j'-1})^{4c}}{2}\geq\frac{(n_{j})^{4c}}{2}$. 
Therefore, in the $k$th round of execution $EX_{j'}$, the good agent labeled $0$ declares having met all good agents and stops, which is incorrect, since it has not met any good agent. This contradicts the definition of algorithm {\tt Alg} and closes the proof of this claim.

Now, consider the largest $x$ such that in each of the $x+1$ first instances $I_0,I_1,\ldots,I_x$ of $\mathcal{I}$, the graph size is at most $n$: these $x+1$ instances constitute family $\mathcal{F}_{n}$. In view of the construction of sequence $\mathcal{I}$ and the definition of $x$, we have $4^{((4c)^x)}\leq n<4^{((4c)^{x+1})}$. Hence, $x$ belongs to $\Omega(\log \log n)$. However, according to Claim~1, the global knowledge given to distinct instances in this family must be different. Hence, there is at least one instance of $\mathcal{F}_{n}$ for which algorithm {\tt Alg} uses a global knowledge of size $\Omega(\log x)$: since $x\in\Omega(\log \log n)$, we have $\Omega(\log x)\in\Omega(\log \log \log n)$. This contradicts the fact that {\tt Alg} uses a global knowledge of size $o(\log \log \log n)$ and proves the theorem.
\end{proof}

%% file: Sections/Conclusion.tex
\section{Conclusion}
\label{sec:ccl}

In this paper, we designed the first polynomial algorithm w.r.t $n$ and $|l_{min}|$ allowing to gather all good agents in presence of Byzantine ones that can act in an unpredictable way and lie about their labels. Our algorithm works under the assumption that the team evolving in the network is strong i.e., the number of good agents is roughly at least quadratic in the number $f$ of Byzantine agents. The required global knowledge $\mathcal{GK}$ is of size $\mathcal{O}(\log \log \log n)$, which is of optimal order of magnitude to get a time complexity that is polynomial in $n$ and $|l_{min}|$ even with strong teams.

A natural open question that immediately comes to mind is to ask if we can do the same by reducing the ratio between the good agents and the Byzantine agents. For example, could it be still possible to solve the problem in polynomial time with a global knowledge of size $\mathcal{O}(\log \log \log n)$ if the number of good agents is at most $o(f^2)$?
Note that the answer to this question may be negative but then may become positive with a little bit more global knowledge. Actually, we can even easily show that the answer is true if the agents are initially given a complete map of the graph with all port numbers, and in which each node $v$ is associated to the list of all labels of the good agents initially occupying node $v$. However, the size of $\mathcal{GK}$ is {then huge} as it belongs to $\Omega(n^2)$. In fact, in this case what is really interesting is to find the optimal size for $\mathcal{GK}$. This observation allows us to conclude with the following open problem that is more general and appealing. 

{\em What are the trade-offs among the ratio good/Byzantine agents, the time complexity and the amount of global knowledge to solve $f$-Byzantine gathering?} 

Bringing an exhaustive and complete answer to this question appears to be really challenging but would turn out to be a major step in our understanding of the problem.